\documentclass{article}
\usepackage[utf8]{inputenc}
\usepackage{float}
\usepackage{jheppub}
\usepackage{amsmath}
\usepackage[english]{babel}
\usepackage{amssymb}
\usepackage{mathtools}
\usepackage{color}
\definecolor{ao}{rgb}{0.0, 0.5, 0.0}
\newcommand{\red}[1]{\textcolor{black}{#1}}

\newcommand{\fd}[1]{\textcolor{black}{#1}}
\usepackage[caption = false]{subfig}
\usepackage{graphicx}
\usepackage{longtable}
\usepackage{tikz}
\usepackage{comment}
\usepackage{capt-of}
\usepackage{amsthm}
\usepackage{hyperref}
\hypersetup{
    colorlinks=true,
    linkcolor=blue
    }

\usepackage{bm}

\newtheorem{remark}{Remark}
\newtheorem{theorem}{Theorem}
\newtheorem{proposition}{Proposition}
\newtheorem{corollary}{Corollary}
\newtheorem{definition}{Definition}
\newtheorem{lemma}{Lemma}

\def\0{\mbox{\tiny $0$}}
\def\1{\mbox{\tiny $1$}}
\def\2{\mbox{\tiny $2$}}
\def\3{\mbox{\tiny $3$}}
\def\4{\mbox{\tiny $4$}}
\def\5{\mbox{\tiny $5$}}
\def\6{\mbox{\tiny $6$}}
\def\7{\mbox{\tiny $7$}}
\def\8{\mbox{\tiny $8$}}
\def\9{\mbox{\tiny $9$}}

\def\r{\rangle}
\def\l{\langle}

\def\R{\mathcal{R}}
\def\b{\beta^{'}}

\def\s{\bm \sigma}

\newcommand{\SOMMA}[2]{\displaystyle\sum\limits_{#1}^{#2}}

\long\def \beq#1\eeq {\begin{equation} #1 \end{equation}}
\long\def \beaq#1\eeaq {\begin{equation}\begin{aligned} #1 \end{aligned}\end{equation}}
\long\def \bes#1\ees {\begin{equation}\begin{split} #1 \end{split} \end{equation}}
\long\def \bea#1\eea {\begin{eqnarray} #1 \end{eqnarray}}
\long\def \bse[#1]#2\ese {\begin{subequations}\label{#1}\begin{align} #2 \end{align}\end{subequations}}

\newcommand{\sums}{\sum_{ \boldsymbol \sigma}}

\newcommand{\si}{\sigma_i}

\title{Hebbian Learning from First Principles}

\author[a,b]{Linda Albanese,}
\author[b,c]{Adriano Barra,}
\author[a]{Pierluigi Bianco,}
\author[a]{Fabrizio Durante,}
\author[a,b]{Diego Pallara,}

\affiliation[a]{Dipartimento di Matematica e Fisica,  Università  del Salento, Lecce, Italy}

\affiliation[b]{Istituto Nazionale di Fisica Nucleare, Sezione di Lecce, Italy}

\affiliation[c]{Dipartimento di Scienze di Base ed Applicate per l’Ingegneria, Sapienza Università di Roma, Italy}

\abstract{Recently, the original storage prescription for the Hopfield model of neural networks -- as well as for its {\em dense} generalizations -- has been turned into a genuine Hebbian learning rule by postulating the expression of its Hamiltonian for both the {\em supervised} and {\em unsupervised} protocols.
\newline
In these notes, first,  we obtain these explicit expressions by relying upon maximum entropy extremization à la Jaynes. Beyond providing a formal derivation of these recipes for Hebbian learning, this construction also highlights how Lagrangian constraints within entropy extremization force network's  outcomes on neural correlations: these try to mimic the  empirical counterparts hidden in the datasets provided to the network for its training and, the denser the network, the longer the correlations that it is able to capture.
\newline
Next, we prove that, in the {\em big data} limit, whatever the presence of a teacher (or its lacking), not only these Hebbian learning rules converge to the original storage prescription of the Hopfield model but also their related free energies (and, thus, the statistical mechanical picture provided by Amit, Gutfreund and Sompolinsky is fully recovered).
\newline
As a sideline, we show mathematical equivalence among standard Cost functions ({\em Hamiltonian}), preferred in Statistical Mechanical jargon, and quadratic Loss Functions, preferred in Machine Learning terminology.
\newline
Remarks on the exponential Hopfield model (as the limit of dense networks with diverging density) and semi-supervised protocols are also provided.}
 
\begin{document}

\maketitle
\section{Introduction}
Since Parisi's groundbreaking discovery of multiple equilibria in complex systems \cite{GiorgioOld,GiorgioNobel},  spin glasses (i.e. {\em complex systems} tackled by statistical mechanical techniques \cite{MPV,FisherHertz}) were immediately recognized as suitable systems for modelling the spontaneous information processing capabilities collectively shown by assemblies of neurons in neural network models \cite{Coolen,Nishimori,Haiping,Engel,Kadmon,LenkaJPA,storage2}. The first celebrated  statistical mechanical analysis of a spin glass model for neural networks was addressed in the eighties by Amit, Gutfreund and Sompolinsky (resulting in the so-called AGS theory) \cite{AGS,Amit}, focusing on the Hopfield model \cite{Hopfield}. The latter is a fully connected mean field neural network whose neurons interact in couples through a synaptic coupling 
that implements the celebrated Hebbian prescription {\em neurons that fire together wire together} \cite{Hebb}.
\newline
Since that milestone, countless variations on theme appeared along these decades and the Hopfield model soon became the {\em harmonic oscillator} for neural networks spontaneously performing pattern recognition as well as the reference for biologically inspired associative memory models. As a natural extension of the Hebbian paradigm to many-body interactions, recently  {\em dense neural networks} have been introduced in the Literature  \cite{HopfieldKrotovEffective,HopKro1,Krotov2018} (see also \cite{Densesterne,Bovier,Fachechi1,DenseRSB}). The latter keep the Hebbian prescription but account for neurons that do not interact any longer in couples, rather in groups larger than two. These have shown to largely outperform the Hopfield scenario both in terms of storage capacity \cite{Baldi,Gardner,Burioni} as well as to perform pattern recognition tasks \cite{BarraPRLdetective,AgliariDeMarzo}\footnote{It is worth mentioning that the limit where the size of these groups diverges is a rather new model, informally called {\em exponential Hopfield model} \cite{demircigil},  whose storage capacity is exponentially larger (in the size of the network) that that of the Hopfield reference, and still by far larger w.r.t. those of all the dense (polynomial) neural networks.}. 
\newline
However, despite its successes in these fields of AI (namely pattern storage \cite{Engel} and pattern recognition \cite{Bishop}), the Hopfield model implements the Hebbian learning prescription solely as a storage criterion: patterns are not inferred from examples, rather they are provided just once to the network and the latter stores them within its synapses à la Hebb. Thus it is by no surprise that, in order to cope with modern Machine Learning problems  (where typically neural networks are trained to datasets in order to extract information and create their own representation of the patterns eventually contained within them), recently the Hebbian prescription of the original Hopfield network, as well as of its dense generalizations, has turned from a simple storing rule to a genuine learning rule, splitting the unsupervised case from the supervised case. See \cite{EmergencySN} for unsupervised learning and \cite{prlmiriam} for supervised learning in the pairwise scenario as well as \cite{unsup,super}, respectively, for the dense extensions.  This shift from {\em Hebbian storing} to {\em Hebbian learning}, however, did not emerge  from {\em first principles}, rather it came out of the blue just by providing intuitive heuristic arguments in order to adjust conveniently the Hopfield Cost function to the case.
\newline
\newline
In this paper, at first, we aim to fill this gap and show how these learning criteria can be obtained from an inferential variational principle, namely the maximum entropy, if used in its interpretation à la Jaynes \cite{Jaynes}: beyond a formal derivation of these learning rules, such an approach allows us to inspect also how and which empirical correlation functions lying in the datasets are tackled by the network during training. In particular, we show that, while the original pairwise network constructs patterns without relying upon higher order statistical moments (i.e. it detects just means and variances, namely one- and two-point correlation functions among the bits of any pattern), in the dense counterpart, the higher the order of the interaction among neurons (i.e. the larger their assembly), the longer the correlations  the network can infer, up to the limit of infinite interactions that approaches to the so-called {\em exponential Hopfield model} \cite{demircigil,Lucibello}. As a sideline, since we provide the expressions for both the unsupervised and supervised learning schemes, we also extend the results to semi-supervised learning (that is, learning from a dataset whose  examples are just partially labelled \cite{SemiSup}).
\newline
Next, by relying upon Guerra's interpolation \cite{GuerraNN, guerra_broken}, we prove that, in the limit of infinite information provided to the network (a scenario that we call {\em big data} limit \cite{super,unsup}), whatever the learning procedure (i.e., with or without the presence of a teacher) these Hebbian learning rules converge to the original Hopfield prescription for pattern storage\footnote{Strictly speaking,  we have convergence to the Hopfield network solely if neural interactions are pairwise, while the dense counterpart converges to the Hebbian P-spin-glass (studied by Elisabeth Gardner in Statistical Mechanics \cite{Gardner} and Pierre Baldi in Computer Science  \cite{Baldi}), as it should be.}. We prove this both at the level of the Cost functions as well as involving the whole free energy. Thus, in this limit, AGS theory \cite{AGS} is entirely recovered. \red{Guerra's interpolation is a mathematical technique introduced for the first time by Francesco Guerra \cite{guerra2002thermodynamic} for spin glasses and later imported to neural networks \cite{AABO-JPA2020, Jean0}. It is mathematically justified 
and, in general, easier than the so-known \textit{replica trick} \cite{MPV}.} Finally, we also prove a one-to-one relation between these new Hamiltonians (i.e., {\em Cost functions}), usually preferred in Statistical Mechanics and standard {\em Loss functions}, typically used in Machine Learning, so to extend the results obtained for the former to the latter. 

\par\medskip
The paper is structured as follows: first the shallow limit of pairwise neural networks is investigated (hence the learning rules for the Hopfield model are derived from the maximum entropy principle and their big data limit is shown to collapse to AGS theory), then we extend these results to dense networks. As an overture, to get acquainted with Hebbian storing, we briefly revise the standard Hopfield model and we derive its Cost function for Hebbian storing via the maximum entropy too\footnote{This exercise is also useful to realize that the only perspective by which entropy maximization should be pursued to picture the properties of the Hopfield neural networks via Statistical Mechanics is by the statistical inference interpretation of the maximum entropy  (as nor these networks are at equilibrium, neither are its constituents, the neurons)  as masterfully deepened by Bialek and coworkers in an experimental and theoretical tour de force along the past two decades \cite{BialekNew,BialekOld,BialekBook}. Furthermore, despite this class of neural network models under investigation is known in the Machine Learning community as {\em energy based models} \cite{LeCunn}, the Hamiltonian of the Hopfield model ({\em vide infra}) is just a Cost function to be minimized but it does not represents a physical energy.}.
 
\section{Preamble: Hebbian storing}
For the sake of completeness, and in order to facilitate comprehension about the way Hebbian learning takes place later on in the manuscript, in this \red{S}ection we briefly revise the celebrated Hopfield model: far from being exhaustive, first  we succinctly and informally introduce the Hamiltonian and the related main observables for its statistical mechanical treatment (Sec. \ref{Sec21}). Then, we derive the Hamiltonian accounting for such a network via the maximum entropy principle from a statistical inference perspective (Sec. \ref{Sec22}).

\subsection{The Hopfield neural network in statistical mechanics}\label{Sec21}

To define the Hebbian prescription for pattern storage, we need at first to introduce the $K$ patterns that we want the network to store and retrieve: these are chosen as $N$-bits random vectors whose entries are drawn from a Rademacher distribution as stated by the next
\begin{definition}\label{DefinitionUno} (Random structureless patterns)
Let $K, N \in \mathbb{N}$. The patterns in the Hopfield model are denoted by $\bm \xi^\mu=(\xi_1^\mu, \hdots, \xi_N^\mu)$, for $\mu=1, \hdots, K$ and their entries are i.i.d. Rademacher random variables with probability mass function given by
\begin{equation} \label{eq:rademacher}
\mathbb{P}(\xi_{i}^{\mu}=x) =  \frac{1}{2} \left ( \delta(x-1) + \delta(x+1) \right),
\end{equation}
for any $i=1, \hdots, N$ and $\mu=1, \hdots, K$.
\end{definition}
Next, we introduce the $N$ neurons and their dichotomous neural activities $\sigma_i \in \{ - 1,+1\}, i =1,...,N$ (such that $\sigma_i =+1$ represents the $i^{th}$ neuron firing, while $\sigma_i =-1$ accounts for its quiescence), that we use  to introduce the Hopfield neural network \cite{Amit}: 
\begin{definition}\label{def:HOP} (Hebbian storing)
Let $J_0 \in \mathbb{R}^+$ and $h \in \mathbb{R}$ and let $\boldsymbol \sigma = ( \sigma_1, \sigma_2, ..., \sigma_N)  \in \{ -1, +1 \}^N$ be a vector of $N \in \mathbb{N}$ binary Ising neurons. The Cost function (or Hamiltonian) of the Hopfield model is defined as
\begin{eqnarray}
    H_N^{\textnormal{(Hop)}}(\s \vert J_0, h, \bm \xi)&=& -\sum_{i,j=1,i<j}^{N}J_{ij}\sigma_i\sigma_j - \sum_{i=1}^N h_i \sigma_i \notag \\  
     &=& -\dfrac{J_0}{2N} \sum_{i,j=1}^{N} \sum_{\mu=1}^K \xi_i^\mu \xi_j^\mu \si \sigma_j -h \sum_{i=1}^N \sum_{\mu=1}^K \xi_i^{\mu}\sigma_i,
    \label{eq:H_Hop}
\end{eqnarray} 
where we have defined $h_i:=h \SOMMA{\mu=1}{K} \xi_i^\mu, \ i=1, \hdots N$.\\
Here, the information about the patterns is allocated in the synaptic matrix, namely in the couplings $\boldsymbol J = \left(J_{ij}\right)_{i,j=1,...,N} \in \mathbb{R}^{N\times N}$ among neurons, defined according to the so-called Hebb's rule \cite{Hopfield}
\begin{equation}\label{eq:hebb}
J_{ij}=\frac{J_0}{N}\sum_{\mu=1}^K \xi_i^{\mu}\xi_j^{\mu}.
\end{equation}
\end{definition}
Note that, in this model, neurons interact in couples, whose strength is ruled by $J_0$, and they perceive the external world, e.g, the presence of a pattern to recognize, through the one-body term (i.e., the second term at the r.h.s. of eq.~\eqref{eq:H_Hop}), whose strength is ruled by $h$.
 
In order to inspect the network's information processing capabilities, we introduce also the concept of {\em control parameters} and {\em order parameters} as follows
\fd{\begin{definition}\label{MattisOverlap} (Control and order parameters)
Given the Hopfield model, \red{the control parameters are} $\beta \in \mathbb{R}^+$\red{,} the noise (see below in eq. \eqref{dinamica}), and $\alpha := K/N \in \mathbb{R}^+$\red{,} the storage capacity of the network.
\newline
Moreover, the order parameters (required to quantify  the success of storing and retrieval by the Hopfield model \eqref{eq:H_Hop}) are the $K$ Mattis magnetizations (or overlaps) denoted by
    \begin{align}
        m_\mu\red{(\bm \sigma \vert \bm \xi)}&= \dfrac{1}{N} \sum_{i=1}^N \xi^\mu_i \si,
        \label{eq:mmu}
    \end{align}
    for $\mu=1, \hdots, K$. 
\end{definition}}
Roughly speaking,  control parameters are parameters that we freely tune in the network (and we expect that --as for the former-- the higher the noise in the network, the harder the computation for its neurons --whatever it is, storage or retrieval-- much as for the latter, because the larger the storage, the harder the task). \red{On the contrary}, order parameters are response functions, that is, they provide a quantitative measure of network's performances as control parameters vary. 
\newline
Note that, working at $h=0$ and $J_0=1$ for simplicity, we can rewrite the Hamiltonian \eqref{eq:H_Hop} as a quadratic form in the Mattis magnetizations:
    \begin{align}\label{EmmeQuadro}
        H_N^{\textnormal{(Hop)}}(\s \vert J_0=1, h=0, \bm \xi)=-\frac{N}{2} \sum_{\mu=1}^K \red{(m_{\mu}(\bm \sigma \vert \bm \xi))^2}.
    \end{align}
Hence it is trivial to realize that the minima of this Cost function are those provided by neural configurations that result in $|m_{\mu}|=1$ for some $\mu$. Those configurations are said to be ``retrieval states'' because they match perfectly the patterns they are meant to retrieve.
Further, these configurations can also be attractors for the neural dynamics that we now introduce: once a Hamiltonian is provided, it is possible to construct a Markov process for the neural dynamics 
$$
\mathbb{P}(\boldsymbol{\sigma}(t+1) \vert \bm \xi)=\sum_{\boldsymbol{\sigma'}}W(\boldsymbol{\sigma},\boldsymbol{\sigma'})\mathbb{P}(\bm{\sigma'}(t)\vert \bm \xi),
$$
where the transition rates $W(\bm{\sigma},\bm{\sigma'})$ are dictated by eq.~\eqref{eq:H_Hop}.  
For instance, a popular implementation of the above neural dynamics (that can be explicitly coded on a computer), reads
\begin{eqnarray}\label{dinamica}
\mathbb{P}(\boldsymbol{\sigma}(t+1) \vert \bm \xi) &=& \prod_{i=1}^N \frac12 \left(1 + \sigma_i(t+1)\tanh(\beta \theta_i(\bm{\sigma}(t)))\right),\\
\theta_i(\bm{\sigma}(t)) &=&  \sum_{j=1}^N J_{ij}\sigma_j(t) + h_i,
\end{eqnarray}
where the term $\theta_i$ is the {\em post synaptic} field acting on neuron $i$ and the parameter $\beta$ tunes the stochasticity in the network. {Moreover,} for $\beta \to 0^+$, the dynamics reduces to a pure random tossing. Instead, for $\beta \to +\infty$, the dynamics steepest descends toward the minima of the Hamiltonian, the latter (in this deterministic limit) playing as the Ljapounov function for the dynamical process \cite{Coolen}. Remarkably, the fact that the synaptic coupling is symmetric, i.e. $J_{ij}=J_{ji}$, is enough for {\em Detailed Balance} to hold. This guarantees that the long time limit of the stochastic process  defined in eq.~\eqref{dinamica} relaxes to the following Boltzmann-Gibbs distribution $\mathbb P(\boldsymbol{\sigma} \vert \bm \xi)$, related to the Hamiltonian \eqref{eq:H_Hop} \cite{Coolen}.
\begin{equation}\label{BGmeasure}
\lim_{t \to \infty}\mathbb{P}(\boldsymbol{\sigma}(t) \vert \bm \xi) = \mathbb{P}(\boldsymbol{\sigma} \vert \bm \xi) = \frac{e^{-\beta H_N^{\textnormal{(Hop)}}(\boldsymbol \sigma | \boldsymbol \xi)}}{\sum_{\bm \sigma}e^{-\beta H_N^{\textnormal{(Hop)}}(\boldsymbol \sigma | \boldsymbol \xi)}}= \frac{e^{ \frac{ N\beta}{2} \sum_{\mu}^K \red{(m_{\mu}(\bm \sigma \vert \bm \xi))^2}}}{Z_N(\boldsymbol \xi)}
\end{equation}
where the normalization factor $Z_N(\boldsymbol \xi)$ is called {\em partition function} in Statistical Mechanics \cite{MPV}.
\newline
Dealing with pattern storage, the interest lies in constructing neural networks with the largest possible storage capacity and in developing techniques to estimate these values as summarized below
\begin{remark}\label{Si2No}
The Hopfield neural network can handle at most $K \propto N^1$ patterns. In particular $K= \alpha N$ and the maximal value of the storage capacity $\alpha$ is $\alpha_c \approx 0.14$ \cite{Amit}.
\end{remark}
It is rather easy to understand at least the order of magnitude of the Hopfield storage, i.e. $K \propto N^1$, by a signal-to-noise argument that we now streamline: as we want the patterns to be configurations that are stable points, i.e. the {\em attractors} for the neural dynamics provided in eq.~\eqref{dinamica}, let us assume that the network already lies in a pattern and checks for its stability. We set now $\boldsymbol{\sigma}(t)=\boldsymbol{\xi}^1$ and $h=0$ (no bias must be added): if $\boldsymbol{\xi}^1$ is an attractor, it must happen that $\sigma_i(t+1)=\sigma_i(t)$ for all $i = 1,...,N$. Since the stochastic process described in eq.~\eqref{dinamica} prescribes $\sigma_i(t+1)=\textrm{sign}[\tanh(\beta\sum_j J_{ij}\sigma_j(t))]$, by substituting in this formula the expression the coupling $J_{ij}$ provided in eq.~\eqref{eq:hebb} and the assumption $\sigma_j(t)=\xi_j^1$, we have that $\sigma_i(t+1)=\textrm{sign}[\tanh(\beta N^{-1}\sum_{j=1}^N \sum_{\mu=1}^K \xi_i^{\mu}\xi_j^{\mu} \xi_j^{\1}]$: this sign is governed by  $N^{-1}\sum_{\mu=1}^K \xi_i^{\mu}\xi_j^{\mu} \xi_j^{\1} \sim  \xi_i^1 \pm \sqrt{K/N} = \xi_i^1 \pm \sqrt{\alpha}$, where the contribution $\xi_i^1$ plays as the {\em signal} while the storage plays as a noise. Thus, as long as the storage $\alpha < 1$, the value of $\sigma_i(t+1)$ is the sign of $\xi_i^1$, that, in turn,  is also the value for $\sigma_i(t)$: hence patterns can play as attractors for neural dynamics provided they are not too many ($\alpha < 1$). Note that, from an Information Theory perspective, when the network reaches a stationary state (i.e., an attractor representing a stored pattern), the final state of the stochastic variable  $\sigma_i$ directly codes for the $i^{th}$ bit  of information  of the corresponding  pattern that it is retrieving (e.g., $\sigma_i = \xi_i^1$ in the example above).
\newline
\newline
The argument of Remark \ref{Si2No} however is rather rough and, in order to have a sharp estimate about the value of the maximal storage $\alpha_c \approx 0.14$, one has to solve for the free energy (that we just define hereafter while we remind to standard textbook for its explicit evaluation \cite{Amit,Coolen}) and inspect its singularities (that depict computational phase transitions and thus automatically detect the maximal storage $\alpha_c$). The free energy, defined as the intensive logarithm of the partition function, is the main observable in Statistical Mechanics and its extremization accounts for the balance between energy minimization and entropy maximization (at given noise $\beta$ and storage $\alpha$), or --in Machine Learning jargon-- it allows to search for the minima of the Cost function under prescription of the maximum entropy, framed as an inferential criterion.
\newline
Its expression is given by the next
\begin{definition}
Let $\beta$ be the noise of the network and  $\alpha$ the storage capacity of the network.
The free energy of the Hopfield model is defined as 
\begin{equation}\label{frienergy}
    \mathcal A(\alpha,\beta) = \lim_{N \to \infty} \mathcal A_{N}(\alpha,\beta), \quad \textnormal{where} \quad \mathcal A_{N}(\alpha,\beta)= \dfrac{1}{N} \mathbb{E}_{\bm \xi} \log Z_N(\boldsymbol \xi),
\end{equation}
where $\mathbb{E}_{\bm \xi}$ is the expectation over the distribution of the patterns.
\end{definition}
The analysis of the free energy is not the focus of the present investigation, we just need to introduce its definition because we aim to prove that the above expression related to the Hebbian storing can be recovered, in suitable limits (that we inspect in Sec. \ref{Guerralike}), also by the free energy stemming from Hebbian learning.  
\newline
Finally we close this introductory Section with a remark explicitly related to Machine Learning. 
\begin{remark}
\label{rem:loss}
As the statistical mechanical formalization of neural networks defines models via their {\em Cost functions} (i.e. Hamiltonians) while in Machine Learning  models are introduced via their {\em Loss functions}, it is instructive to prove that these objects are one-to-one related.  
\newline
Indeed,  for the Hopfield neural network, if we introduce the $K$ Average $L^2$ Loss functions for reconstruction (one per pattern) as 
\begin{align}
    L_{\pm}^\mu(\bm \sigma \vert \bm \xi) =& \dfrac{1}{2N}  \lVert \bm \xi^{\mu} \pm \bm \sigma \rVert^2_2,
\end{align}
it is enough to note that $L_{\pm}^\mu (\bm \sigma \vert \bm \xi)= 1 \pm m_{\mu}$ to rewrite the Hamiltonian provided in eq.~\eqref{EmmeQuadro} as
\begin{align}
H_N^{\textnormal{(Hop)}}(\boldsymbol \sigma | J_0=1, h=0, \boldsymbol \xi)  &= -\dfrac{N}{2} \sum_{\mu=1}^K m_\mu^2= -\dfrac{N}{2} \sum_{\mu=1}^K (1-\mathcal{L}_N^\mu(\bm \sigma \vert \bm \xi)),
\end{align}
where 
\begin{align}
    \mathcal{L}_N^\mu(\bm \sigma \vert \bm \xi)= L_{-}^\mu (\bm \sigma \vert \bm \xi) L_{+}^\mu(\bm \sigma \vert \bm \xi).
\end{align}
Note the presence of both $L_{-}^\mu$ and $L_{+}^\mu$ in $\mathcal{L}_N^\mu$: this  mirrors the fact that the Hopfield Hamiltonian is a quadratic form in $m_{\mu}$ (hence it can store and reconstruct both $\xi^{\mu}$ as well as $-\xi^{\mu}$), see eq.~\eqref{EmmeQuadro}. 
\end{remark}

\subsection{The Hopfield neural network from statistical inference}\label{Sec22}

In this Section we aim to explore the Maximum Entropy Principle (MEP) to obtain the Boltzmann-Gibbs measure  (see eq.~\eqref{BGmeasure}) related to the Hopfield Hamiltonian defined in eq.~\eqref{eq:H_Hop} from an inferential procedure rather than as the relaxation of a stochastic process. 
\newline
In an experimental scenario, to check retrieval performances of the neural network, we should measure at least two numbers per pattern: the mean values of the overlap between the network's configuration (once relaxed toward a steady state) and the stored pattern, i.e. the Mattis magnetizations, and their variances. 
\newline
In other words, the experimental setup requires the observation of the following quantities
\begin{equation}\label{constraints}
\langle m_\mu \rangle _{{\textnormal{exp}}}= \frac1N \sum_{i=1}^N \xi^\mu _i \langle \sigma_i \rangle _{{\textnormal{exp}}},\qquad \langle m_\mu^2 \rangle _{{\textnormal{exp}}} = \frac1{N^2} \sum_{i,j=1}^{N} \xi^\mu _i\xi^\mu _j \langle \sigma_i\sigma_j \rangle_{\textnormal{exp}},
\end{equation}
where the subscript ${\textnormal{exp}}$ means that we are considering experimentally evaluated quantities. Data collection should happen as follows: upon presentation of the new pattern, say $\bm \xi^k$, the network is left free to evolve toward relaxation within  a stable configuration.  Then,  skipped the initial transient, neural activity is recorded over time $t =1,...,T$ and suitably averaged such that 
$$
\langle m_{k} \rangle_{{\textnormal{exp}}} =\frac{1}{N}\sum_i^N \xi_i^k \langle \sigma_i \rangle_{{\textnormal{exp}}} = \lim_{T \to \infty}\frac{1}{T}\sum_{t=1}^T \frac{1}{N}\sum_i^N \xi_i^k\sigma_i(t).
$$
As discussed in Remark \ref{Si2No}, if the network is handling a moderate number of patterns, then the empirical average is informative as $\langle \sigma_i \rangle_{\textnormal{exp}} = \xi^k_i$.  
\newline
Further, to set the maximum entropy optimization procedure, we need to rely also on the theoretical outcomes from the network for means and variances of the Mattis magnetizations and these are achieved via standard Boltzmann-Gibbs sampling and indicated by the brackets (i.e., $\langle m_{\mu}\rangle$ and $\langle m_{\mu}^2\rangle$) that read as
\begin{align*}
\langle m_{k}\rangle =& \mathbb{E}_{\bm \xi}\frac{\sum_{\bm \sigma}  (\frac{1}{N}\sum_{i=1}^N \xi_i^{k}\sigma_i) e^{-\beta H_N^{\textnormal{(Hop)}}(\s \vert J_0, h, \bm \xi)}}{\sum_{\bm \sigma} e^{-\beta H_N^{\textnormal{(Hop)}}(\s \vert J_0, h, \bm \xi)}}, \notag \\
\langle m^2_{k}\rangle =& \mathbb{E}_{\bm \xi} \frac{\sum_{\bm \sigma} (\frac{1}{N^2}\sum_{i,j=1}^{N} \xi_i^{k} \xi_k^{k} \sigma_i \sigma_j) e^{-\beta H_N^{\textnormal{(Hop)}}(\s \vert J_0, h, \bm \xi)}}{\sum_{\bm \sigma} e^{-\beta H_N^{\textnormal{(Hop)}}(\s \vert J_0, h, \bm \xi)}}.
\end{align*}
The brackets represent the Boltzmann average over the ensemble $\omega(.)$ and over the quenched patterns $\xi$, namely --given a generic observable $O(\sigma, \xi)$ (function of both the neural activities and the patterns lying in the synaptic coupling)-- we write
$$
\langle O(\sigma, \xi) \rangle := \mathbb{E}\omega\left(O(\sigma, \xi) \right) =
\mathbb{E}\frac{\sum_{\sigma} O(\sigma, \xi) e^{-\beta H_N^{Hop}(\boldsymbol{\sigma}|J_0,h,\boldsymbol{\xi})}}{Z_N(\boldsymbol{\xi})}\equiv \mathbb{E}\frac{\sum_{\sigma} O(\sigma, \xi) e^{-\beta H_N^{Hop}(\boldsymbol{\sigma}|J_0,h,\boldsymbol{\xi})}}{\sum_{\sigma} e^{-\beta H_N^{Hop}(\boldsymbol{\sigma}|J_0,h,\boldsymbol{\xi})}}.
$$
The goal is then to determine the less structured probability distribution $\mathbb{P}_N(\boldsymbol{\sigma} \vert \boldsymbol{\xi})$ accounting for these data, namely such that $\langle m_{\mu} \rangle =\langle m_{\mu} \rangle_{\textnormal{exp}}$ and $\langle m_{\mu}^2 \rangle = \langle  m_{\mu}^2 \rangle_{\textnormal{exp}}$, for all $\mu =1,...,K$. To do so, \fd{following the Maximum Entropy Principle à la Jaynes, we state}

\begin{proposition}
Given $K$ patterns $\{ \bm \xi^{\mu} \}_{\mu=1,...,K}$ generated as prescribed by  Def. \ref{DefinitionUno} and $N$ Ising neurons, whose activities read as $\sigma_i \in \{ - 1, +1\}$ for all $i =1,...,N$, the less structured probability distribution $\mathbb{P}_N(\boldsymbol{\sigma} \vert \boldsymbol{\xi})$, whose first and second order's moments  $\langle m_{\mu} \rangle$ and $\langle m_{\mu}^2 \rangle$ reproduce the empirical counterparts $\langle m_{\mu} \rangle_{\textnormal{exp}}$ and $\langle m_{\mu}^2 \rangle_{\textnormal{exp}}$, reads as 
\begin{equation}
\mathbb{P}_N(\boldsymbol{\sigma} \vert \boldsymbol{\xi}) = \frac{e^{-\beta H_N^{(Hop)}(\bm \sigma \vert J_0, h, \boldsymbol{\xi})}}{Z(\boldsymbol{\xi})} =\frac{\exp\left(\frac{\beta J_0}{2N}\sum_{i,j=1}^{N} \sum_{\mu=1}^K \xi^\mu _i\xi^\mu _j  \sigma_i\sigma_j + \beta h \sum_{i=1}^{N}  \sum_{\mu=1}^{K}\xi^\mu _i \sigma_i\right)}{\sum_{\boldsymbol{\sigma}} \exp\left(\frac{\beta J_0}{2N}\sum_{i,j=1}^{N} \sum_{\mu=1}^K \xi^\mu _i\xi^\mu _j  \sigma_i\sigma_j + \beta h \sum_{i=1}^{N}  \sum_{\mu=1}^{K}\xi^\mu _i \sigma_i\right)},
\end{equation}
and its Cost function can be interpreted as the Hamiltonian of the Hopfield model, see \eqref{eq:H_Hop}.  This result is obtained by maximizing the  constrained Shannon entropy defined in eq.~\eqref{shannonH} ({\em vide infra}).
\end{proposition}

\begin{proof}
The key argument is to obtain $\mathbb{P}_N(\boldsymbol{\sigma} \vert \boldsymbol{\xi})$ by maximizing the Shannon entropy 
$$
S[\mathbb{P}_N(\boldsymbol{\sigma} \vert \boldsymbol{\xi})]=-\sum_{\boldsymbol{\sigma}} \mathbb{P}_N(\boldsymbol{\sigma} \vert \boldsymbol{\xi})\log \mathbb{P}_N(\boldsymbol{\sigma} \vert \boldsymbol{\xi})
$$  
constraining the latter, via Lagrange multipliers, to have lowest order momenta $\langle m_{\mu} \rangle$ and $\langle m^2_{\mu} \rangle$  that reproduce the empirical ones (i.e.,  $\langle m_{\mu} \rangle_{\textnormal{exp}}$ and $\langle m^2_{\mu} \rangle_{\textnormal{exp}}$).
\newline
We should thus maximize the quantity 
\begin{equation}\label{shannonH}
\begin{split}
S_{\lambda,h,J_0}[\mathbb{P}_N(\boldsymbol{\sigma} \vert \boldsymbol{\xi})]=&-\sum_{\boldsymbol{\sigma}} \mathbb{P}_N(\boldsymbol{\sigma} \vert \boldsymbol{\xi})\log \mathbb{P}_N(\boldsymbol{\sigma} \vert \boldsymbol{\xi}) + \lambda N\Big(\sum _{\boldsymbol{\sigma}}\mathbb{P}_N(\boldsymbol{\sigma} \vert \boldsymbol{\xi})-1\Big)\\
&+ h \beta N \sum_{\mu=1}^K \Big(\sum_{\boldsymbol{\sigma}} \mathbb{P}_N(\boldsymbol{\sigma} \vert \boldsymbol{\xi})\frac1N\sum_{i=1}^N \xi^\mu _i  \sigma_i -\langle m_\mu \rangle_{\textnormal{exp}} \Big)\\
&+\frac{J_0 \beta N}{2} \sum_{\mu=1}^K \Big(\sum_{\boldsymbol{\sigma}} \mathbb{P}_N(\boldsymbol{\sigma} \vert \boldsymbol{\xi})\frac1{N^2} \sum_{i,j=1}^{N} \xi^\mu _i\xi^\mu _j  \sigma_i\sigma_j -\langle m_\mu ^2\rangle_{\textnormal{exp}} \Big),
\end{split}
\end{equation}
with respect to $\mathbb{P}_N(\boldsymbol{\sigma} \vert \boldsymbol{\xi})$ and the Lagrangian multipliers  $\lambda,h, J_0$\footnote{The factor $N$ coupled to all the Lagrangian multipliers in \eqref{shannonH} is due to the fact that the Shannon entropy is linearly extensive in $N$, i.e. $S[\mathbb{P}_N(\boldsymbol{\sigma} \vert \boldsymbol{\xi})] \sim O(N^1)$. Further, just in this subsection, the letter chosen for the Lagrange multipliers (e.g. $h \beta, \ J_0 \beta$) are to preserve the notation of Statistical Mechanics.}. The constraint $\partial_{\lambda} S_{\lambda,h,J_0}[\mathbb{P}_N(\boldsymbol{\sigma} \vert \boldsymbol{\xi})]=0$ is equivalent to requiring that $\mathbb{P}_N(\boldsymbol{\sigma} \vert \boldsymbol{\xi})$ is actually a probability distribution, while $\partial_h S_{\lambda,h,J_0}[\mathbb{P}_N(\boldsymbol{\sigma} \vert \boldsymbol{\xi})] = 0$ and $\partial_{J_0} S_{\lambda,h,J_0}[\mathbb{P}_N(\bm \sigma \vert \bm \xi)]=0$ effectively fix the theoretical observables to the experimental quantities ($J_0$ accounting for the variances and $h$ for the biases, in agreement with their role in Statistical Mechanics, where they tune the two-body and one-body interactions, see eq. \eqref{eq:H_Hop}). Finally,
\begin{equation}
\frac{\delta S[\mathbb{P}_N(\boldsymbol{\sigma} \vert \boldsymbol{\xi})]}{\delta \mathbb{P}_N(\boldsymbol{\sigma} \vert \boldsymbol{\xi})}=- \log \mathbb{P}_N(\boldsymbol{\sigma} \vert \boldsymbol{\xi})-1 + \beta h \sum_{i=1}^N\sum_{\mu=1}^K \xi^\mu _i \sigma_i +\frac\beta{2N}\sum_{i,j=1}^{N}\sum_{\mu}^K \xi^\mu _i\xi^\mu _j  \sigma_i\sigma_j=0,
\end{equation}
which means that
\begin{eqnarray}
\mathbb{P}_N(\boldsymbol{\sigma} \vert \boldsymbol{\xi})={\text{const}} \cdot \exp\left(\frac{\beta J_0}{2N}\sum_{i,j=1}^{N} \sum_{\mu=1}^K \xi^\mu _i\xi^\mu _j  \sigma_i\sigma_j+\beta h \sum_{i=1}^{N}\sum_{\mu=1}^K\xi^\mu _i \sigma_i\right).
\end{eqnarray}
By putting the constant equal to ${\text{const}}= Z_N(\boldsymbol{\xi})^{-1}$, we have the desired assertion. 
\end{proof}

Clearly, the first term in the exponential at the r.h.s. of the above expression is the two-body contribution where synapses store {in Hebbian way} the patterns while the second term is a bias driving the network toward one specific attractor. Note the role that the Lagrange multipliers acquire when moving from Statistical Inference to Statistical Mechanics.
\newline
In the following we will restrict (with no loss of generality) the study of Hebbian learning to the two-body coupling, i.e. we assume $h=0$ for all the neurons  (as the field driving toward a pattern does not affect the free energy landscape where patterns lie).  Further, still with no loss of generality, we keep $J_0=1$ for the sake of simplicity.

\section{Main theme: Hebbian learning}
As discussed in the Introduction, despite the expression of the coupling provided in eq.~\eqref{eq:hebb} is often named Hebbian {\em learning}, the above model has little to share with modern Machine Learning as there is no real learning process underlying the storage of patterns within the synaptic matrix of the Hopfield neural network. \red{On the contrary}, mimicking the Machine Learning literature, we would only provide the network with examples of the patterns (and never the patterns themselves) in order to check if and how, once enough examples have been supplied, the network is able to reconstruct the patterns hidden in the information it experienced. 
To do so, we need also to generalize the synaptic coupling in the r.h.s. of eq.~\eqref{eq:H_Hop} in order to let the network experience examples of patterns, rather than pattern themselves. 
Let us first generate such a dataset of examples.

Starting from the previous K patterns, whose distribution is defined in  eq.~\eqref{eq:rademacher}, we aim to create a collection of $M$ noisy versions for each of them, using the following prescription:
\fd{\begin{definition} (Random structureless dataset)
\label{def:example}
Let $M \in \mathbb{N}$ and $K \in \mathbb{N}$. Given the pattern $\bm \xi^\mu$, made of $N$ elements, we define the set of examples as the dataset $\{ \boldsymbol \eta^{\mu,a} \}_{a=1,...,M}^{\mu=1,...,K}$ obtained by $M$ randomly perturbed copies of each pattern whose generic entry $(i,\mu,a)$ is distributed as
\begin{equation} \label{eq:Bernoulli}
    {\mathbb{P}(\eta_{ i}^{\mu,a}|\xi^\mu_i) = \frac{1-r}{2} \delta_{\eta_{ i}^{\mu,a},-\xi_{ i}^{\mu}} + \frac{1+r}{2} \delta_{\eta_{ i}^{\mu,a},\xi_{ i}^{\mu}}}.
\end{equation}
\end{definition}}
\fd{Thus, $r \in [0,1]$ assesses the training-set \emph{quality}, that is, as $r \rightarrow 1$ the example matches perfectly the related pattern, whereas for $r \to 0$ the example is 
orthogonal to the related pattern as $N \to +\infty$. $M \in \mathbb{N}$, instead, accounts for the training-set \emph{quantity}.}

We can now use this dataset to generalize the synaptic coupling in the Hopfield model and turn it into a learning machine. This is achieved, so far, just as a working ansatz, by providing the following two definitions (one per protocol) of the generalized Hamiltonians to the case  \cite{prlmiriam,EmergencySN}:

\begin{definition}\label{def:sup}{(Supervised Hebbian learning)}
    Given \fd{the examples} $\{ \bm \eta^{\mu,a} \}_{a=1,...,M}^{\mu=1,...,K}$ generated as prescribed in Def. \ref{def:example} and $N$ Ising neurons, whose activities read as $\sigma_i  \in \{ -1, +1\}$ for all $i =1,...,N$, the Cost function (or {\em Hamiltonian}) of the Hebbian neural network in the supervised regime is \fd{defined as}
    \begin{align}
    H_{N,M}^{\textnormal{(Sup)}}(\s \vert \bm \eta)=-\dfrac{1}{2N\R}\sum_{\mu=1}^K \sum_{i,j=1, i<j}^{N} \left( \dfrac{1}{M} \sum_{a=1}^{M} \eta_i^{\mu,a} \right)\left( \dfrac{1}{M} \sum_{b=1}^{M} \eta_j^{\mu,b} \right) \si \sigma_j
    \label{eq:H_sup}
\end{align}
where $\R:= r^2 + \dfrac{1-r^2}{M}$ is a normalization factor.
\end{definition}

\begin{definition}{(Unsupervised Hebbian learning)}
\label{def:unsup}
    Given \fd{the examples} $\{ \bm \eta^{\mu,a} \}_{a=1,...,M}^{\mu=1,...,K}$ generated as prescribed in Def. \ref{def:example} and $N$ Ising neurons, whose activities read as $\sigma_i \in \{ - 1, +1\}$ for all $i =1,...,N$, the Cost function (or {\em Hamiltonian}) of the Hebbian neural network in the unsupervised regime is 
    \begin{align}
    \label{eq:H_unsup}
    H_{N,M}^{\textnormal{(Uns)}}(\s\vert \bm \eta) = -\dfrac{1}{2N\mathcal{R}M} \sum_{\mu=1}^K \sum_{i,j=1, i<j}^{N} \sum_{a=1}^M  \eta_i^{\mu,a} \eta_j^{\mu,a}\si \sigma_j.
\end{align}
where $\R:= r^2 + \dfrac{1-r^2}{M}$ is a normalization factor\footnote{The value $\R$ corresponds to the variance of the random variable $\dfrac{1}{M} \SOMMA{a=1}{M} \eta_i^{\mu,a}$.}.
\end{definition}

\fd{Note that, contrary to the supervised protocol, in the unsupervised counterpart all the examples are provided at once to the network }(regardless the pattern they pertain to, i.e. $J_{ij} \propto \SOMMA{a=1}{M} \eta_i^{\mu,a}\eta_j^{\mu,a}$), in the supervised scenario a teacher knows which example belongs to which pattern and splits them accordingly (such that the resulting synaptic matrix is much more informative, i.e., $J_{ij} \propto \left(\SOMMA{a=1}{M} \eta_i^{\mu,a}\right)\left(\SOMMA{b=1}{M} \eta_j^{\mu,b}\right)$).


Mirroring the Hebbian storing, see Def. \ref{MattisOverlap}, we need to introduce the following 
\begin{definition}{(Control and order parameters)}
Beyond the noise in the network $\beta \in \mathcal{R}^+$ and the storage capacity of the network $\alpha:=K/N \in \mathcal{R}^+$ (as in the previous model at work solely with pattern storage), further control parameters for Hebbian learning are  also the dataset quality $r$ and the dataset quantity $M$: note that these two variables merge together to form $\rho= (1-r^2)/Mr^2$ that is the entropy of the dataset\footnote{Stricktly speaking, $\rho := (1-r^2)/Mr^2$ is the conditional entropy of the pattern given the examples, i.e., $S(\xi_i^{\mu}|\bm{\eta_i^{\mu}})$. The latter accounts for the amount of information required to describe the pattern entry $\xi_i^{\mu}$ given the knowledge of $M$ independent observations of (examples of) the pattern  $\bm{\eta}_i^{\mu}$ $=(\eta_i^{\mu,a=1},...,\eta_i^{\mu,a=M})$. While we deepened this point elsewhere (see e.g. \cite{prlmiriam}), to see that $\rho$ quantifies the conditional entropy of the pattern given the examples $S(\xi_i^{\mu}|\bm{\eta_i^{\mu}})$, an intuitive argument is to note that the error probability per bit is $\mathbb{P}(\chi_i^{\mu,a}=-1)=(1-r)/2$, thus by applying the majority rule to the empirical vector $\bm{\eta}_i^{\mu}$ we get $\mathbb{P}(\textrm{sign}(\sum_{a=1}^M \chi_i^{\mu,a}=-1) \sim 1 - erf(1/\sqrt{2\rho})$. Further, it is simple to check that the dataset is maximally informative when its entropy is null, namely $\rho=0$, that is or $r=1$ (i.e. examples are perfect copies of the patterns) or $M \to \infty$ (i.e. examples are infinite).}.
\newline
Beyond the K Mattis magnetizations of the patterns, the natural order parameters of the Hebbian neural network in the unsupervised regime are the $K \times M$ generalized Mattis magnetizations
\begin{align}\label{PdO-Uns}
n_{\mu,a}&= \dfrac{1}{N}\sum_{i=1}^N \eta_i^{\mu,a} \si,
\end{align}
for each $\mu=1, \hdots ,K$ and $a=1, \hdots, M$ \cite{EmergencySN}.  In the supervised regime, instead, beyond the K Mattis magnetizations of the patterns, the generalized Mattis magnetizations (related to the examples) are defined as 
\begin{align}\label{PdO-Sup}
\tilde n_\mu&= \dfrac{1}{M N} \sum_{a=1}^M \sum_{i=1}^N \eta_i^{\mu,a} \si,
\end{align}
for each $\mu=1, \hdots ,K$ \cite{prlmiriam}.
\label{def:ordparam}
\end{definition}
It is a trivial exercise to prove that the Cost functions provided in Def. \ref{def:sup} and Def. \ref{def:unsup} can indeed be written as quadratic forms, respectively, in the order parameters provided by eq.~\eqref{PdO-Uns} and eq.~\eqref{PdO-Sup}. Further, it is elementary to realize that --when moving from storing to learning-- other questions beyond the maximal storage capacity of the network can be addressed, at first the thresholds for learning, namely the minimal amount of information that must be provided to the network in order for it to correctly infer the patterns: these computational aspects are known and they have been discussed elsewhere \cite{prlmiriam,Lucibello}.

\subsection{Maximum entropy principle for shallow neural networks}
Plan of this Section is two-fold: the main one is to derive the (so far assumed) expressions for the Cost functions for the two regimes of Hebbian learning provided in Def.s \ref{def:sup} and \ref{def:unsup} by relying upon the maximum entropy extremization. Further,  by inspecting the way we  must constraint the entropy functional to reach this goal, we can also better understand which kind of correlations the network searches in the dataset to form its own probabilistic representation of the patterns. 
\newline
Once addressed these points, also semi-supervised learning (that is training a network with a mixed dataset built of by labelled and unlabelled examples) is briefly discussed. 
\newline
\newline
The MEP \cite{Jaynes} allows to determine the less structured probability distribution whose moments  best fit  their empirical counterparts\footnote{Note that the empirical moments have to be experimentally measured as for instance discussed for the derivation of the Hopfield storage prescription, see Sec. \ref{Sec22}. See also, e.g. \cite{Application1,Application2,Application3} for concrete applications of this method to real datasets collected in biological sciences (neuroscience, ecology and immunology respectively).}. As we are going to show in the first part of this manuscript, as long as we apply this technique to pairwise networks (as the Hopfield model where neurons interact in couples),  these moments turn out to be just the first-order moments (i.e. means and variances or one-point and two-point correlation functions).  \red{On the contrary}, in the second part of this manuscript, we show how, extending the theory to many body interactions, higher order moments have to be accounted.

\subsubsection{Supervised Hebbian learning}
\label{sec:SUP_MEP}
We deal with the supervised scheme first: let us start with the following 
\begin{definition}
\label{def:empir_sup}
    In the (shallow) supervised Hebbian learning, once provided a dataset as defined in Def. \ref{def:example}, with $\R$ as defined in Def. \ref{def:sup}, $N$ Ising neurons, whose activities read as $\sigma_i \in \{- 1, +1\}$ for all $i =1,...,N$,  the  empirical averages of $\tilde n_\mu$ and $\tilde n_\mu^2$ that we need to collect from the dataset are denoted by 
\begin{align}
    \langle \tilde n_\mu \rangle_{\textnormal{exp}} &= \dfrac{r}{\R NM} \sum_{i=1}^N \sum_{a=1}^M \eta_i^{\mu,a} \langle \si \rangle_{\textnormal{exp}}, \\ \label{eq:2.4_1}
    \langle \tilde n_\mu^2 \rangle_{\textnormal{exp}} &= \left(\dfrac{r}{\R NM}\right)^2 \sum_{i,j=1}^{N}\sum_{a,b=1}^{M} \eta_i^{\mu,a}\eta_j^{\mu,b} \langle \si \sigma_j \rangle_{\textnormal{exp}}.
\end{align}
\end{definition}
\fd{Notice that,} as previously remarked, the presence of a teacher allows to write the summation over the examples in the above eq.~\eqref{eq:2.4_1} as performed on two separate indices $(a,b)$.

\begin{proposition}
\label{propMEP}
    Given a dataset $\{ \bm \eta^{\mu,a} \}_{a=1,...,M}^{\mu=1,...,K}$ generated as prescribed in Def. \ref{def:example} and $N$ Ising neurons, whose activities read as $\sigma_i \in \{-1,+1\}$ for all $i=1,...,N$, the less structured probability distribution  $ \mathbb{P}_{N,M}(\s \vert \bm \eta)$ whose first and second order's moments reproduce the empirical averages fixed in Definition \ref{def:empir_sup} reads as 
    \begin{align}
    \label{eq:propMEP_sup}
        \mathbb{P}_{N,M}(\s \vert \bm \eta) =& \dfrac{1}{Z(\bm \eta)} \exp \left(
        \sum_{\mu=1}^K  \dfrac{\lambda_1^\mu r}{\R NM} \sum_{i=1}^N\sum_{a=1}^{M} \eta_i^{\mu,a} \si  
        +\sum_{\mu=1}^K \lambda_2^\mu \left(\dfrac{r}{\R NM}\right)^2 \sum_{i,j=1}^{N}\sum_{a,b=1}^{M} \eta_i^{\mu,a}\eta_j^{\mu,b}\si \sigma_j\right),
    \end{align}
    where $Z(\bm \eta)$ is the partition function, $\lambda_1^\mu$ and $\lambda_2^\mu$, $\mu = 1,...,K$ are the Lagrangian multipliers. In particular, $\lambda_1^\mu=0$ and $\lambda_2^\mu=\dfrac{\beta \R N}{2 r^2}$, for $\mu=1, \hdots, K$, are the parameters that reproduce the details of the supervised Hebbian learning scheme\footnote{$\lambda_1=0$ simply reflects that the distribution of the examples is symmetric. \\ The independence of $\lambda_2$ by the index $\mu$ reflects that the weights allocated to each pattern (i.e. the amplitudes of their basins of attraction) are all the same.}.
\end{proposition}
 
\begin{remark}
A direct consequence of the mathematical expression of the Shannon entropy, in particular the presence of a logarithm in its definition, is that the solution of the MEP can always be seen as a Boltzmann-Gibbs probability distribution of a given Hamiltonian, namely
$$
\mathbb{P}_{N,M}(\s \vert \bm \eta) = \frac{e^{-\beta H_{N,M}(\boldsymbol{\sigma}|\boldsymbol{\eta})}}{Z(\boldsymbol{\eta})}, 
$$
thus it is immediate to verify that $H_{N,M}(\boldsymbol{\sigma}|\boldsymbol{\eta})=H^{(Sup)}_{N,M}(\boldsymbol{\sigma}|\boldsymbol{\eta})$ by comparing the Cost function provided in  Definition \ref{def:sup} and the exponent at the r.h.s. of eq.~\eqref{eq:propMEP_sup} with the values of the Lagrange multipliers selected in Proposition \ref{propMEP}.
\end{remark}
\begin{proof}\textit{(Proposition \ref{propMEP})}
We apply the maximum entropy principle that, from a formal perspective, is nothing but \fd{reduces to} a constrained optimization problem w.r.t. the distribution $\mathbb{P}_{N,M}(\s \vert \bm \eta)$, solvable through method of Lagrangian multipliers. The Lagrange entropic functional is 
\begin{align}
    S[ \mathbb{P}_{N,M}(\s \vert \bm \eta)]=& - \sum_{\s} \mathbb{P}_{N,M}(\s \vert \bm \eta) \log \mathbb{P}_{N,M}(\s \vert \bm \eta) + \lambda_0 \left( \sum_{\s} \mathbb{P}_{N,M}(\s \vert \bm \eta) - 1\right) \notag \\
    &+ \sum_{\mu=1}^K \lambda_1^\mu \left( \sum_{\s} \mathbb{P}_{N,M}(\s \vert \bm \eta)\dfrac{r}{\R NM} \sum_{i=1}^N\sum_{a=1}^M \eta_i^{\mu,a} \si -  \langle \tilde n_\mu \rangle_{\exp}\right) \notag \\
    &+ \sum_{\mu=1}^K \lambda_2^\mu \left( \sum_{\s} \mathbb{P}_{N,M}(\s \vert \bm \eta)\left(\dfrac{r}{\R NM}\right)^2 \sum_{i,j=1}^{N}\sum_{a,b=1}^{M} \eta_i^{\mu,a}\eta_j^{\mu,b}\si \sigma_j - \langle \tilde n_\mu^2 \rangle_{\exp}\right)
\end{align}
where the first term at the r.h.s. represents the entropy of the system and the Lagrangian multipliers clamp $\mathbb{P}_{N,M}(\s \vert \bm \eta)$ to be a probability distribution (this job is done by $\lambda_0$) whose mean values of $\tilde n_\mu$ and $\tilde n_\mu^2$ reproduce their corresponding experimental ones (these matching are guaranteed by $\lambda_1^{\mu}$ and $\lambda_2^{\mu}$ respectively). \\
The stationarity conditions of $S[\mathbb{P}_{N,M}(\s \vert \bm \eta)]$ w.r.t. $\mathbb{P}_{N,M}(\s \vert \bm \eta)$, $\lambda_0$, $\lambda_1^\mu$ and $\lambda_2^\mu$ yield
\begin{align}
\label{eq:distrMEP}
\begin{cases}
    \dfrac{\partial S[\mathbb{P}_{N,M}(\s \vert \bm \eta)]}{\partial \lambda_0} =& \sum_{\s} \mathbb{P}_{N,M}(\s \vert \bm \eta) - 1 =0\\ 
    \dfrac{\partial S[\mathbb{P}_{N,M}(\s \vert \bm \eta)]}{\partial \lambda_1^\mu} =&  \sum_{\s} \mathbb{P}_{N,M}(\s \vert \bm \eta)\dfrac{r}{\R NM} \SOMMA{i=1}{N}\SOMMA{a=1}{M} \eta_i^{\mu,a} \si -  \langle \tilde n_\mu \rangle_{\exp} =0 \\
     \dfrac{\partial S[\mathbb{P}_{N,M}(\s \vert \bm \eta)]}{\partial \lambda_2^\mu}=&\sum_{\s} \mathbb{P}_{N,M}(\s \vert \bm \eta)\left(\dfrac{r}{\R NM}\right)^2 \SOMMA{i,j=1}{N}\SOMMA{a,b=1}{M} \eta_i^{\mu,a}\eta_j^{\mu,b}\si \sigma_j - \langle \tilde n_\mu^2 \rangle_{\exp} =0\\
    \dfrac{\delta S[\mathbb{P}_{N,M}(\s \vert \bm \eta)]}{\delta \mathbb{P}_{N,M}} =& -\log \mathbb{P}_{N,M}(\s \vert \bm \eta)-1 +\lambda_0 + \SOMMA{\mu=1}{K} \lambda_1^\mu \dfrac{r}{\R NM} \SOMMA{i=1}{N}\SOMMA{a=1}{M} \eta_i^{\mu,a} \si \notag \\
    &+ \SOMMA{\mu=1}{K} \lambda_2^\mu \left(\dfrac{r}{\R NM}\right)^2 \SOMMA{i,j=1}{N}\sum_{a,b=1}^{M} \eta_i^{\mu,a} \eta_j^{\mu,b}  \sigma_i \sigma_j =0     
    \end{cases}
    \end{align}
    The first three equations give us the condition of $\mathbb{P}_{N,M}(\s \vert \bm \eta)$ to be a discrete probability distribution and the connection between the theoretical averages of $\tilde n_\mu$ and $\tilde n_\mu^2$ and their experimental counterparts. The last equation, instead, returns  the explicit expression of the probability distribution that results from the extremization procedure, namely   
    \begin{align}
    \log \mathbb{P}_{N,M}(\s \vert \bm \eta) =& {\left( \lambda_0 -1 +  \SOMMA{\mu=1}{K} \lambda_1^\mu \dfrac{r}{\R NM} \SOMMA{i=1}{N}\SOMMA{a=1}{M} \eta_i^{\mu,a} \si + \SOMMA{\mu=1}{K} \lambda_2^\mu \left(\dfrac{r}{\R NM}\right)^2 \SOMMA{i,j=1}{N}\SOMMA{a,b=1}{M} \eta_i^{\mu,a} \eta_j^{\mu,b} \sigma_i \sigma_j\right)}.
\end{align}
If we fix the normalization by setting  $Z(\bm \eta)=\exp\left( 1-\lambda_0 \right)$ we recover a Boltzmann-Gibbs distribution and, in particular,  by choosing\footnote{$\lambda_1^\mu=0$ simply reflects that the distribution of examples is symmetric. The independence of $\lambda_2^\mu$ by the index $\mu$ reflects that the weights allocated to each pattern (i.e. the amplitudes of their basins of attraction) are all the same.} $\lambda_1^\mu=0$\ and $\lambda_2^\mu=\dfrac{\beta \R N}{2 r^2}$, for $ \mu=1, \hdots, K$, the parameters reproduce the details of the supervised Hebbian learning scheme.
\end{proof}

\subsubsection{Unsupervised Hebbian learning}
\label{sec:UNSUP_MEP}
Mirroring the structure of subsection \ref{sec:SUP_MEP}, we use maximum entropy principle now to recover the Boltzmann-Gibbs expression of the probability distribution $\mathbb{P}_{N,M}(\s \vert \bm \eta)$ accounting for the unsupervised protocol (i.e. whose Cost functions matches that of shallow unsupervised Hebbian learning provided in Def. \ref{def:unsup}). 
To this task, we need to premise the following 
\begin{definition}
\label{def:empir}
    In the (shallow) unsupervised Hebbian learning, once provided a dataset $\{ \bm \eta^{\mu,a} \}_{a=1,...,M}^{\mu=1,...,K}$ as defined in Def. \ref{def:example}, with $\R$ as introduced in Def. \ref{def:unsup}, the empirical averages of $n_{\mu,a}$ and $n_{\mu,a}^2$ that we need to collect from the dataset read as 
\begin{align}
    \langle n_{\mu,a} \rangle_{\textnormal{exp}} &= \dfrac{r}{\R N} \sum_{i=1}^N \eta_i^{\mu,a} \langle \si \rangle_{\textnormal{exp}},\\ \label{eq:2.4}
    \langle n_{\mu,a}^2 \rangle_{\textnormal{exp}} &= \left(\dfrac{r}{\R N}\right)^2 \sum_{i,j=1}^{N} \eta_i^{\mu,a}\eta_j^{\mu,a} \langle \si \sigma_j \rangle_{\textnormal{exp}}, 
\end{align}
where the index $a$ has not been saturated since we lack a teacher now that provides this information.
\end{definition}

Now we state the main proposition of this subsection: 
\begin{proposition}
\label{propMEP_UNSUP}
    Given a dataset $\{ \bm \eta^{\mu,a} \}_{a=1,...,M}^{\mu=1,...,K}$ generated as prescribed in Def. \ref{def:example} and $N$ Ising neurons, whose activities read as $\sigma_i \in \{ -1, +1\}$ for all $i =1,...,N$,  the less structured probability distribution $\mathbb{P}_{N,M}(\s \vert \bm \eta)$ whose first and second order's moments reproduce the empirical averages fixed in Definition \ref{def:empir} reads as 
    \begin{align}
        \mathbb{P}_{N,M}(\s \vert \bm \eta) =& \dfrac{1}{Z(\bm \eta)} \exp \left( 
        \sum_{\mu=1}^K\sum_{a=1}^{M} \lambda_1^{\mu,a} \dfrac{r}{\R N} \sum_{i}^N \eta_i^{\mu,a} \si 
        +\sum_{\mu=1}^K\sum_{a=1}^M \lambda_2^{\mu,a} \left(\dfrac{r}{\R N}\right)^2 \sum_{i,j=1}^{N} \eta_i^{\mu,a} \eta_j^{\mu,a} \sigma_i \sigma_j\right)
        \label{eq:probMEP_UNSUP}
    \end{align}
    where $Z(\bm \eta)$ is the partition function, $\lambda_1^{\mu,a}$ and $\lambda_2^{\mu,a}$  for all $\mu = 1,...,K$ and $a = 1,...,M$ are Lagrangian multipliers. In particular, $\lambda_1^\mu=0$ and $\lambda_2^\mu=\dfrac{\beta \R N}{2 M r^2}$, for $\mu=1, \hdots, K$, are the parameters that reproduce the details of the unsupervised Hebbian learning scheme.
\end{proposition}

\begin{proof}\textit{(Proposition \ref{propMEP_UNSUP})}
We follow the path paved to prove the analogous result within the supervised setting. 
The Lagrange entropic functional, associated to the probability distribution $\mathbb{P}_{N,M}(\s \vert \bm \eta)$ and constrained to reproduce the empirical lowest order statistics provided in Def. \ref{def:empir}, reads as 
\begin{align}
    S[\mathbb{P}_{N,M}(\s \vert \bm \eta)]=& - \sum_{\s} \mathbb{P}_{N,M}(\s \vert \bm \eta) \log \mathbb{P}_{N,M}(\s \vert \bm \eta) + \lambda_0 \left( \sum_{\s} \mathbb{P}_{N,M}(\s \vert \bm \eta) - 1\right) \notag \\
    &+ \sum_{\mu=1}^K\sum_{a=1}^M \lambda_1^{\mu,a} \left( \sum_{\s} \mathbb{P}_{N,M}(\s \vert \bm \eta)\dfrac{r}{RN} \sum_{i=1}^N \eta_i^{\mu,a} \si -  \langle n_{\mu,a} \rangle_{\exp}\right) \notag \\
    &+ \sum_{\mu=1}^K\sum_{a=1}^M \lambda_2^{\mu,a} \left( \sum_{\s} \mathbb{P}_{N,M}(\s \vert \bm \eta)\left(\dfrac{r}{RN}\right)^2 \sum_{i,j=1}^{N} \eta_i^{\mu,a}\eta_j^{\mu,a}\si \sigma_j - \langle n_{\mu,a}^2 \rangle_{\exp}\right)
\end{align}
where the first term represents the Shannon entropy of the network while the Lagrangian multipliers force $\mathbb{P}_{N,M}(\s \vert \bm \eta)$ to be a probability distribution (via $\lambda_0$) and clamp the mean values of $\langle n_{\mu,a}\rangle$ and $\langle n_{\mu,a}^2 \rangle$ to be the experimental ones (via $\lambda_1^{\mu}$ and $\lambda_2^{\mu}$ respectively). \\
The stationarity conditions of $S[\mathbb{P}_{N,M}(\s \vert \bm \eta)]$ w.r.t. $\mathbb{P}_{N,M}(\s \vert \bm \eta)$, $\lambda_1^\mu$ and $\lambda_2^\mu$ for $\mu =1,...,K $ yield
\begin{align}
\label{eq:distrMEP2}
\begin{cases}
    \dfrac{\partial S[\mathbb{P}_{N,M}(\s \vert \bm \eta)]}{\partial \lambda_0} =& \sum_{\s} \mathbb{P}_{N,M}(\s \vert \bm \eta) - 1 =0\\ 
    \dfrac{\partial S[\mathbb{P}_{N,M}(\s \vert \bm \eta)]}{\partial \lambda_1^{\mu,a}} =&  \sum_{\s} \mathbb{P}_{N,M}(\s \vert \bm \eta)\dfrac{r}{\R N} \SOMMA{i=1}{N} \eta_i^{\mu,a} \si -  \langle n_{\mu,a} \rangle_{\exp}=0 \\
     \dfrac{\partial S[\mathbb{P}_{N,M}(\s \vert \bm \eta)]}{\partial \lambda_2^{\mu,a}}=&\sum_{\s} \mathbb{P}_{N,M}(\s \vert \bm \eta)\left(\dfrac{r}{\R N}\right)^2 \SOMMA{i,j=1,1}{N,N} \eta_i^{\mu,a}\eta_j^{\mu,a}\si \sigma_j - \langle n_{\mu,a}^2 \rangle_{\exp}=0\\
    \dfrac{\delta S[\mathbb{P}_{N,M}(\s \vert \bm \eta)]}{\delta \mathbb{P}_{N,M}(\s \vert \bm \eta)} =& -\log \mathbb{P}_{N,M}(\s \vert \bm \eta) -1 +\lambda_0 + \SOMMA{\mu=1}{K} \SOMMA{a=1}{M} \lambda_1^{\mu,a} \dfrac{r}{\R N} \SOMMA{i=1}{N} \eta_i^{\mu,a} \si \notag \\
    &+ \SOMMA{\mu=1}{K} \SOMMA{a=1}{M} \lambda_2^{\mu,a} \left(\dfrac{r}{\R N}\right)^2 \SOMMA{i,j=1}{N} \eta_i^{\mu,a} \eta_j^{\mu,a} \sigma_i \sigma_j =0     
    \end{cases}
    \end{align}
    The first three equations return the requirement for $\mathbb{P}_{N,M}(\s \vert \bm \eta)$ to be a discrete probability distribution and guarantee the matching of the theoretical averages of $n_{\mu,a}$ and $n_{\mu,a}^2$ to their experimental counterparts. The last equation, instead, gives us the explicit expression of the probability distribution, namely   
    \begin{align}
    \log \mathbb{P}_{N,M}(\s \vert \bm \eta) =& {\left( \lambda_0 -1 + \SOMMA{\mu=1}{K} \SOMMA{a=1}{M}\lambda_1^{\mu,a} \dfrac{r}{RN} \SOMMA{i=1}{N} \eta_i^{\mu,a} \si+ \SOMMA{\mu=1}{K} \SOMMA{a=1}{M} \lambda_2^{\mu,a} \left(\dfrac{r}{RN}\right)^2 \SOMMA{i,j=1}{N} \eta_i^{\mu,a}  \eta_j^{\mu,a} \sigma_i \sigma_j\right)}.
\end{align}
If we fix the normalization by setting  $Z(\bm \eta)=\exp\left( 1-\lambda_0 \right)$ we recover a Boltzmann-Gibbs distribution and, in particular,  by requiring $\lambda_1^{\mu,a}=0$ and $\lambda_2^{\mu,a}=\dfrac{\R N}{2 r^2 M}$, for all $\mu =1,...,K$, $\mathbb{P}_{N,M}(\s \vert \bm \eta)$ reproduces all the details of the unsupervised Hebbian learning scheme.
\end{proof}

\subsubsection{Semi-supervised Hebbian learning}  
In this subsection we analyze a middle way between the settings of the last two subsections, known in Literature as \textit{semi-supervised learning} \cite{SemiSup}: in this approach the neural network is kept the same, simply the dataset is partially labelled and partially not (namely the teacher helped in grouping together only a subset of all the examples).  Before we start, we premise the following 
\begin{remark}\label{MixedData} 
    Let us consider now $M$ examples split into two disjoint categories whose sizes are $M_1$ and $M_2$ respectively and
    let us introduce a parameter $s \in [0,1]$ such that $s=M_1/M$ and $1-s=M_2/M$ and the two variables $\rho_{sup}$ and $\rho_{uns}$, accounting for the entropies of the labelled and unlabelled parts of the dataset, respectively, as 
    \begin{align}
        \rho_{sup}&= \dfrac{1-r^2}{r^2 M_1}= \dfrac{\rho}{s} \notag \\
        \rho_{uns}&= \dfrac{1-r^2}{r^2 M_2} = \dfrac{\rho}{1-s}.
    \end{align}    
\end{remark}

Moreover, we recall the definition of the order parameters $\tilde{n}_{\mu}$ and $n_{\mu,a}$ in Def. \ref{def:ordparam} such that we can finally write the next

\begin{definition} (Semi-supervised Hebbian learning)
 Given the dataset provided in Remark \ref{MixedData}, the Cost function (or Hamiltonian) of the Hebbian  neural network in the semi-supervised learning setting reads as 
    \begin{align}
  \label{eq:H}
        H_{N,M}^{(Semi)}(\boldsymbol{\sigma} \vert s,\ \bm \eta)= -\dfrac{N}{2}\dfrac{1}{r^2 M_2}\SOMMA{a=1}{M_2}\SOMMA{\mu=1}{K}\left[\dfrac{s}{\sqrt{(1+\rho_{s})}}\tilde n_\mu +\dfrac{(1-s)}{\sqrt{(1+\rho_{us})}}n_{\mu,a}\right]^2
    \end{align}
    \label{def:semisup}
\end{definition}
\begin{remark}
It is a trivial exercise to check that in the limit of a fully labelled (unlabelled) dataset, the Cost function of the supervised (unsupervised) Hebbian learning are recovered, i.e.
\begin{eqnarray}
    H_{N,M}^{(Semi)}(\boldsymbol{\sigma} \vert s=1,\ \bm \eta) &=& H_{N,M}^{(Sup)}(\boldsymbol{\sigma} \vert \bm \eta),\\
    H_{N,M}^{(Semi)}(\boldsymbol{\sigma} \vert s=0,\ \bm \eta) &=& H_{N,M}^{(Uns)}(\boldsymbol{\sigma} \vert \bm \eta).
\end{eqnarray}
\end{remark}
We will not address the computational properties of a pairwise neural network performing semi-supervised learning (these can be found elsewhere, see for instance \cite{SemiSup}), rather we aim to apply the maximum entropy prescription to recover the Boltzmann-Gibbs probability distribution related to this model and, thus, its Cost function. \\
Let us start from the following

\begin{definition}
\label{def:empaverSS}
    In the (shallow) semi-supervised Hebbian learning, the empirical average of $n_{\mu,a}$, $\tilde{n}_\mu$, $n_{\mu,a}^2$, $\tilde n_\mu^2$ and $n_{\mu,a}\tilde n_\mu$, with $a=1, \hdots, M_2$, that we need to collect from the dataset, are 
    \begin{align}
    \langle n_{\mu,a} \rangle_{\textnormal{exp}} &= \dfrac{r}{N \R} \sum_{i=1}^N \eta_i^{\mu,a} \langle \si \rangle_{\textnormal{exp}}, \\ 
    \langle n_{\mu,a}^2 \rangle_{\textnormal{exp}} &= \left(\dfrac{r}{N \R}\right)^2 \sum_{i,j=1}^{N} \eta_i^{\mu,a}\eta_j^{\mu,a} \langle \si \sigma_j \rangle_{\textnormal{exp}}, \\
        \langle \tilde n_\mu \rangle_{\textnormal{exp}} &= \dfrac{r}{N \R M_1} \sum_{i= 1}^{N} \sum_{b=1}^{M_1} \eta_i^{\mu,b} \langle \si \rangle_{\textnormal{exp}}, \\
    \langle \tilde n_\mu^2 \rangle_{\textnormal{exp}} &= \left(\dfrac{r}{N \R M_1}\right)^2 \sum_{i,j=1}^{N}\sum_{b,c=1}^{M_1} \eta_i^{\mu,c}\eta_j^{\mu,b} \langle \si \sigma_j \rangle_{\textnormal{exp}}, \\
    \langle \tilde n_\mu n_{\mu,a} \rangle_{\textnormal{exp}} &= \dfrac{r^2}{N^2\R^2 M_1} \sum_{i,j=1}^{N}\sum_{b=1}^{M_1} \eta_i^{\mu,a}\eta_j^{\mu,b} \langle \si \sigma_j \rangle_{\textnormal{exp}}. 
\end{align}
\end{definition}

\begin{proposition}
     Given a dataset $\{ \bm \eta^{\mu,a} \}_{a=1,...,M}^{\mu=1,...,K}$ generated as prescribed in Def. \ref{def:example} and $N$ Ising neurons, whose activities read as $\sigma_i \in \{ -1, +1\}$ for all $i =1,...,N$,  the less structured probability distribution $\mathbb{P}_{N,M}(\s \vert \bm \eta, s)$ whose first and second order's moments reproduce the empirical ones as fixed in Def. \ref{def:empaverSS} reads as 
    \begin{align}
        \mathbb{P}_{N,M}(\s \vert \bm \eta, s) =& \dfrac{1}{Z(\bm \eta)} \exp \left( \sum_{\mu=1}^K\sum_{a=1}^M \lambda_1^{\mu,a} \dfrac{r}{\R N} \sum_{i=1}^N \eta_i^{\mu,a} \si + \sum_{\mu=1}^K\sum_{a=1}^M \lambda_2^{\mu,a} \left(\dfrac{r}{N \R}\right)^2 \sum_{i,j=1}^{N} \eta_i^{\mu,a} \si\eta_j^{\mu,a} \sigma_j \right. \notag \\
        &+ \sum_{\mu=1}^K \lambda_1^\mu \dfrac{r}{ N \R M_1} \sum_{i=1}^N \sum_{a=1}^{M_1} \eta_i^{\mu,a} \si + \sum_{\mu=1}^K \lambda_2^\mu \left(\dfrac{r}{ N \R M_1}\right)^2 \sum_{i,j=1}^{N}\sum_{a,b=1}^{M_1} \eta_i^{\mu,a} \eta_j^{\mu,b} \si \sigma_j \notag \\
        &\left.+ \sum_{\mu=1}^K \sum_{a=1}^M  \tilde \lambda_2^{\mu,a} \dfrac{r^2}{ N^2 \R^2 M_1}\sum_{i,j=1}^{N}\sum_{b=1}^{M_1} \eta_i^{\mu,a} \eta_j^{\mu,b} \si \sigma_j\right)
    \end{align}
\end{proposition}

The proof of the Proposition mirrors those provided for Propositions \ref{propMEP} and \ref{propMEP_UNSUP} hence we omit it. We only stress that by choosing 
\begin{align}
    &\lambda_1^{\mu,a}=\lambda_1^\mu=0, 
    & &\lambda_2^\mu= \dfrac{N \R s^2}{2r^2(1+\rho_s)}, \notag \\
    &\lambda_2^{\mu,a}= \dfrac{N \R (1-s)^2}{2r^2(1+\rho_{us})},
    & &\tilde \lambda_2^{\mu,a}= \dfrac{N \R s(1-s)}{2r^2 \sqrt{(1+\rho_s)(1+\rho_{us})}}.
\end{align}
we recover the expression of the Boltzmann-Gibbs probability distribution related to the semi-supervised setting as prescribed by the  Hamiltonian provided in eq.~\eqref{eq:H}.

\begin{corollary}
    If $s=0$ we recover the Boltzmann-Gibbs probability distribution of the unsupervised protocol  
        \begin{equation}
        \mathbb{P}_{N,M}(\s \vert \bm \eta, s=0)= \dfrac{1}{Z(\bm \eta)} \exp \left(
        \sum_{\mu=1}^K\sum_{a=1}^M \lambda_2^{\mu,a} \left(\dfrac{r}{RN}\right)^2 \sum_{i,j=1}^{N} \eta_i^{\mu,a} \eta_j^{\mu,a} \sigma_i \sigma_j\right);
    \end{equation}
    instead, by selecting $s=1$, the Boltzmann-Gibbs probability distribution  of the supervised protocol is returned
        \begin{equation}
        \mathbb{P}_{N,M}(\s \vert \bm \eta, s=1) = \dfrac{1}{Z(\bm \eta)} \exp \left(
        \sum_{\mu=1}^K \lambda_2^\mu \left(\dfrac{r}{\R NM}\right)^2 \sum_{i,j=1}^{N}\sum_{a,b=1}^{M} \eta_i^{\mu,a}\eta_j^{\mu,b} \sigma_i \sigma_j\right).
    \end{equation}
\end{corollary}



\subsection{Recovering the free energy of Hebbian storage in the big data limit}\label{Guerralike}

Plan of this Section is to prove that the whole free energy of these learning protocols approaches the free energy of the Hopfield model in the $M\to\infty$ limit (the so-called {\em big data} limit), thus, as a consequence all the related properties of that model are recovered \cite{Amit}. To do so, we rely upon Guerra's interpolation technique \cite{GuerraNN}. The idea is to write a generalized free energy that interpolates between the free energies of Hebbian learning and Hebbian storing and then prove that its derivative w.r.t. the interpolation parameter goes to zero as $M$ goes to infinity. This allows a straightforward application of the Fundamental Theorem of Calculus that makes the proof extremely simple.
\newline
As a sideline, we also perform the big data limit to the Boltzmann-Gibbs probability distributions stemming by these learning rules to prove that, in the $M\to\infty$ limit, these approach  the Boltzmann-Gibbs measure corresponding to standard Hopfield storage: this implies that these Cost functions for Hebbian learning  ultimately collapse to the Hopfield Hamiltonian for Hebbian storing, as intuitive in the big data limit.

\subsubsection{Supervised Hebbian learning}
Let us start showing how the free energy of Hebbian learning, within the supervised setting, approaches the free energy of the Hopfield model, in the big data limit. 
\newline
It is useful to write explicitly the next
\begin{definition} (Free energy)
The free energy of the (shallow) supervised Hebbian learning, related to the Cost function provided in Definition \ref{def:sup},  at finite sizes $N,M$ reads as
\begin{equation}\label{Free1Super}
\mathcal{A}_{N,M}(\alpha,\beta) = \dfrac{1}{N} \mathbb{E} \log \sum_{\s} \exp\left( \dfrac{\beta}{2N\R}\sum_{\mu=1}^K \sum_{i,j=1,1}^{N,N} \left( \dfrac{1}{M} \sum_{a=1}^{M} \eta_i^{\mu,a} \right)\left( \dfrac{1}{M} \sum_{b=1}^{M} \eta_j^{\mu,b} \right) \si \sigma_j\right).
\end{equation}
where $\mathbb{E}$ is the average w.r.t. $\bm \eta$, $\bm \eta = \bm \eta (\xi, r)$.
\end{definition}
\fd{We also remind that the free energy of the Hopfield model at finite size $N$ is 
\begin{equation}\label{RemindSmemo}
    \mathcal{A}_N(\alpha, \beta)= \dfrac{1}{N} \mathbb{E} \log \sums \exp\left( \dfrac{\beta}{2N} \sum_{\mu=1}^K \sum_{i,j=1}^N \xi_i^\mu \xi_j^\mu \si \sigma_j\right),
\end{equation}
where $\mathbb{E}$ is now the average w.r.t. $\bm \xi$ (see eq.~\eqref{frienergy}).}

We can now introduce the following interpolating free energy:
\begin{definition}{(Interpolating free energy)}\\   
Considering $t \in [0,1]$, the K patterns as introduced in Def. \ref{DefinitionUno}, their related noisy dataset $\{ \bm \eta^{\mu,a} \}_{a=1,...,M}^{\mu=1,...,K}$ as prescribed in Def. \ref{def:example}, the Cost functions of Hebbian storing (see eq.~\eqref{eq:H_Hop}) and supervised Hebbian learning (see eq.~\eqref{eq:H_sup}), the interpolating free energy is introduced as 
\begin{align}\label{interpo}
    \mathcal{A}_{N,M}(t\vert \alpha,\beta)= \dfrac{1}{N} \mathbb{E} \log \sum_{\s} &\exp \left(t\dfrac{\beta}{2N} \sum_{i,j=1}^{N} \sum_{\mu=1}^K \xi_i^\mu \xi_j^\mu \si \sigma_j\right. \notag \\
    &\left.+ (1-t)\dfrac{\beta}{2N\R}\sum_{\mu=1}^K \sum_{i,j=1}^{N} \left( \dfrac{1}{M} \sum_{a=1}^{M} \eta_i^{\mu,a} \right)\left( \dfrac{1}{M} \sum_{b=1}^{M} \eta_j^{\mu,b} \right) \si \sigma_j\right) \notag \\
    &=\dfrac{1}{N} \mathbb{E}\log Z_{N,M}(t \vert \bm \xi, \bm \eta) = \dfrac{1}{N} \mathbb{E}\log \sum_{\s} B(t \vert \s, \bm \xi, \bm \eta).
\end{align}
where $\R$ is provided in Def. \ref{def:unsup}, $Z_{N,M}(t \vert \bm \xi, \bm \eta)$ is the interpolating partition function, $B(t \vert \s, \bm \xi, \bm \eta)$ the related Boltzmann factor and $\mathbb{E}$ is the expectation w.r.t the distribution of $\bm \xi$ and $\bm \eta$.
\end{definition}
We stress that, for $t=0$, the above expression recovers the expression of free energy related to the supervised Hebbian learning (see  eq.~\eqref{Free1Super}) , whereas for  $t=1$ it reproduces the expression of the free energy of the Hopfield model (see  eq.~\eqref{RemindSmemo}) . 

\begin{remark}
    A generalized average follows from this generalized measure linked to the Boltzmann factor $B(t \vert \bm \sigma, \bm \xi, \bm \eta)$ as
\beq
	\omega_{t} (\cdot) \coloneqq  \frac{1}{ Z_{N, M}( t \vert \bm \xi, \boldsymbol{\eta})} \, \sum_{\boldsymbol \sigma}~ \cdot ~   B_{N, M} 
 ( t \vert \s, \bm \xi, \boldsymbol{\eta})
	\eeq
	and we preserve the following notation for the brackets
\beq
\langle \cdot   \rangle  \coloneqq \mathbb E [ \omega_{t} ( \cdot) ].
\eeq
\end{remark}
Note that $\omega_t(.)$ returns the Boltzmann average of the Hebbian storing when evaluated at $t=1$ while it returns the Boltzmann average of the supervised Hebbian learning when evaluated at $t=0$. 
\newline
With these premises, we can now state the main Theorem of this subsection as
\begin{theorem}
\label{prop:limM_sup}
In the asymptotic (\em{big data}) limit, $M \to +\infty$, the expression of the free energy related to the supervised Hebbian learning approaches that of the Hopfield model.
\end{theorem}

To prove the Theorem, we need to put beforehand the following 
\begin{lemma}
\label{lemma:dert1}
In the limit $M \to +\infty$, the derivative of the interpolating free energy \eqref{interpo} w.r.t. $t$ is null, namely 
\begin{align}
   \lim_{M \to \infty} \dfrac{d\mathcal{A}_{N,M}(t\vert \alpha,\beta)}{dt} = \lim_{M \to \infty}  \dfrac{\beta}{2}\sum_{\mu=1}^K \left( \langle m_\mu^2 \rangle - \dfrac{1}{\R}\langle \tilde n_\mu^2 \rangle \right)= 0.
\end{align}
\end{lemma}
\begin{proof}
    Let us start with the computation of the derivative: 
    \begin{align}
        d_t \mathcal{A}_{N,M}(t\vert \alpha,\beta) =& \dfrac{1}{N} \mathbb{E}\Bigg[\dfrac{1}{ Z_{N,M}(t \vert \s, \bm \xi, \bm \eta)}\sum_{\s} B(t \vert \s, \bm \xi, \bm \eta)\left( \dfrac{\beta}{2N} \sum_{i,j=1}^{N} \sum_{\mu=1}^K \xi_i^\mu \xi_j^\mu \si \sigma_j\right.\notag \\
        &\left.-\dfrac{\beta}{2N\R}\sum_{\mu=1}^K \sum_{i,j=1}^{N} \left( \dfrac{1}{M} \sum_{a=1}^{M} \eta_i^{\mu,a} \right)\left( \dfrac{1}{M} \sum_{b=1}^{M} \eta_j^{\mu,b} \right) \si \sigma_j\right)\Bigg] \notag \\
        =& \dfrac{\beta}{2 N} \sum_{\mu=1}^K \left(\langle m_\mu^2 \rangle - \dfrac{1}{\R} \langle \tilde n_\mu^2 \rangle  \right).
        \label{eq:deriv}
    \end{align}
    We apply the \red{Central Limit Theorem (CLT)} to the random variable $\dfrac{1}{M} \SOMMA{a=1}{M} \eta_i^{\mu,a}$ in the $M \to +\infty$ limit. With some trivial manipulations we can state that 
\begin{align}
    \mathbb{E} \left(\dfrac{1}{M}\sum_{a=1}^M \eta_i^{\mu,a} \Big \vert \bm \xi\right)&= \xi^\mu_i \\
    \mathbb{E}_{\bm \xi} \left(\dfrac{1}{M}\sum_{a=1}^M \eta_i^{\mu,a}\right)^2 &= r^2+ \dfrac{1-r^2}{M} = \R.
\end{align}
Therefore, we can express $\dfrac{1}{M}\SOMMA{a=1}{M} \eta_i^{\mu,a}$ as a Gaussian variable with null average and whose variance reads $\R/M$, such that, as $M \to \infty$, its probability distribution becomes a Dirac's delta peaked in $\xi_i^\mu$. Hence, in eq.~\eqref{eq:deriv}, in the $M \to +\infty$ limit, we can replace $\dfrac{1}{M} \sum_{a=1}^M \eta_i^{\mu,a}$ with $\R \xi_i^\mu$, so we get the definition of $m_\mu$ (see eq.~\eqref{eq:mmu}) and, thus, we reach the thesis. 
    \end{proof}
    Now we are ready to prove the Theorem.
\begin{proof}(of Theorem \ref{prop:limM_sup})
We apply the fundamental theorem of calculus over the interpolating variable $t$ while $M \to \infty$: 
\begin{align}
    \mathcal{A}_{N,M}(t=1\vert\alpha,\beta)= \mathcal{A}_{N,M}(t=0\vert \alpha,\beta) + \int_0^1 \dfrac{d \mathcal{A}_{N,M}(t\vert \alpha,\beta)}{dt}\vert_{t=s} ds,
\end{align}
where the commutativity of the $M\to\infty$ limit and the integral over the interpolation parameter is justified by dominated convergence. Since for Lemma \ref{lemma:dert1} we have that, in $M \to \infty$ limit, the derivative of the interpolating free energy w.r.t. $t$ becomes zero, this implies 
\begin{align}
    \mathcal{A}_{N}(t=1\vert\alpha,\beta) = \mathcal{A}_{N}(t=0\vert\alpha,\beta),
\end{align}
so the equivalence is proved. 
\end{proof}
Note that this result holds at finite network sizes $N$, hence the reason for the subscript $N$ in $\mathcal{A}_{N}$.

\begin{corollary}
\label{cor:S_SUP}
In the asymptotic limit of {\em big data}, namely for $M \to +\infty$,  the Boltzmann-Gibbs  probability distribution $\mathbb{P}_{N,M}(\s \vert \bm \eta)$,  related to the supervised Hebbian learning Cost function (provided by eq.~\eqref{eq:propMEP_sup}), approaches the Boltzmann-Gibbs  distribution of the Hopfield model, namely
\begin{equation}
    \lim_{M \to \infty} \mathbb{P}_{N,M}(\s \vert \bm \eta) = \mathbb{P}_{N}(\s \vert \bm \xi) = \frac{1}{Z(\bm \xi)}\exp \left( 
    \dfrac{\beta}{2N}\sum_{\mu=1}^K \sum_{i, j=1}^{N} \xi_i^{\mu}   \xi_j^{\mu} \sigma_i \sigma_j\right).
    \label{eq:Hebb_stor}
\end{equation} 
\end{corollary}

\begin{proof}
Focusing on $\mathbb{P}_{N,M}(\s \vert \bm \eta)$ in  eq.~\eqref{eq:propMEP_sup} and, in particular, by applying the \red{CLT} to the random variable $\dfrac{1}{M}\SOMMA{a=1}{M} \eta_i^{\mu,a}$ in the $M \to +\infty$ limit, as already done in Lemma \ref{lemma:dert1}, its probability distribution becomes a Dirac's delta peaked in $\xi^\mu_i$ and, thus, we reach the thesis. 
\end{proof}
We highlight the plain agreement, in the \em{big data} limit, among outcomes of Theorem \ref{prop:limM_sup} and those of Corollary \ref{cor:S_SUP} because, as the Cost function of Hebbian learning collapses on the Cost function of the Hopfield model for pattern storage (Corollary \ref{cor:S_SUP}), by no surprise their related free energies become equal (Theorem \ref{prop:limM_sup}), yet we stress that the methods underlying the proofs of these two results are independent, the former based on the Fundamental Theorem of Calculus, the latter on concentration of measure induced by CLT arguments.

\subsubsection{Unsupervised Hebbian learning}
Purpose of this Section is to prove that also the free energy (and, in turn, the Boltzmann-Gibbs probability distribution too) related to the Hebbian unsupervised protocol, in the \em{big data} limit, approaches that of the corresponding Hebbian storage by the Hopfield model. Since the strategy we aim to use is quite the same as the previous one, we report only the main definitions and results. 
\begin{definition} (Free energy)
The free energy of the (shallow) unsupervised Hebbian learning, related to the Cost function provided in Definition \ref{def:unsup},  at finite sizes $N,M$ reads as as\begin{equation}\label{freeunsup}
\mathcal{A}_{N,M}(\alpha,\beta) =  \dfrac{1}{N} \mathbb{E} \log \sum_{\s} \exp \left(\dfrac{\beta}{2N\R M}\sum_{\mu=1}^K \sum_{i,j=1}^{N} \sum_{a=1}^{M} \eta_i^{\mu,a} \eta_j^{\mu,a} \si \sigma_j\right),   
\end{equation}
where $\mathbb{E}$ is the average w.r.t. $\bm \eta$. 
We also remind that the free energy of the Hopfield model at finite size $N$ is 
\begin{equation}\label{Doppiona}
    \mathcal{A}_N(\alpha, \beta)= \dfrac{1}{N} \mathbb{E} \log \sums \exp\left( \dfrac{\beta}{2N} \sum_{\mu=1}^K \sum_{i,j=1}^N \xi_i^\mu \xi_j^\mu \si \sigma_j\right),
\end{equation}
where $\mathbb{E}$ is now the average w.r.t. $\bm \xi$.
\end{definition}
We can thus introduce the interpolating free energy as stated by the next
\begin{definition}{(Interpolating free energy)}\\
Considering $t \in [0,1]$, the K patterns as introduced in Def. \ref{DefinitionUno}, their related noisy dataset $\{ \bm \eta^{\mu,a} \}_{a=1,...,M}^{\mu=1,...,K}$ generated as prescribed in Def. \ref{def:example}, the Cost functions of Hebbian storing (see eq.~\eqref{eq:H_Hop}) and of unsupervised Hebbian learning (see eq.~\eqref{eq:H_unsup}), the interpolating free energy is introduced as
\begin{align} 
    \mathcal{A}_{N,M}(t\vert \alpha,\beta) = \dfrac{1}{N} \mathbb{E} \log \sum_{\s} &\exp \left(t\dfrac{\beta}{2N} \sum_{i,j=1}^{N} \sum_{\mu=1}^K \xi_i^\mu \xi_j^\mu \si \sigma_j\right. \notag \\
    &\left.+ (1-t)\dfrac{\beta}{2N\R M}\sum_{\mu=1}^K \sum_{i,j=1}^{N} \sum_{a=1}^{M} \eta_i^{\mu,a} \eta_j^{\mu,a} \si \sigma_j\right) \notag \\
    &\hspace{-2.2cm}=\dfrac{1}{N} \mathbb{E}\log Z_{N,M}(t \vert \bm \xi, \bm \eta)
\end{align}
where $\R$ is provided in Def. \ref{def:unsup} and $Z_{N,M}(t \vert  \bm \xi, \bm \eta)$ is the interpolating partition function. 
\end{definition}
We stress that, for $t=0$, the above interpolating free energy reduces to its expression related to unsupervised Hebbian learning  (see  eq.~\eqref{freeunsup}), whereas for  $t=1$ it approaches the expression related to Hebbian storing, namely that of the Hopfield model (see eq.~\eqref{Doppiona}).

\begin{theorem}
\label{prop:limM_UNSUP}
In the asymptotic (\em{big data}) limit, i.e., $M \to +\infty$, the expression of the free energy related to the unsupervised Hebbian learning approaches that of the Hopfield model.
\end{theorem}
Mirroring the previous subsection, the proof of Theorem \ref{prop:limM_UNSUP} is trivial provided the next
\begin{lemma}
\label{lemma:dert_UNSUP}
In the limit $M \to +\infty$, the derivative of the interpolating free energy w.r.t. $t$ is null, namely 
\begin{align}
    \dfrac{d\mathcal{A}_{N,M}(t\vert \alpha, \beta)}{dt} = \dfrac{\beta}{2}\sum_{\mu=1}^K \left( \langle m_\mu^2 \rangle - \dfrac{1}{\R M} \sum_{a=1}^M\langle n_{\mu,a}^2 \rangle \right) \xrightarrow{M\to \infty} 0.
\end{align}
\end{lemma}
\begin{proof}
As already done for the supervised protocol, we can apply the Fundamental Theorem of Calculus on the interpolating free energy and reach the thesis of Theorem \ref{prop:limM_UNSUP}.
\newline
We only highlight the application of CLT in this case: if we focus on the random variables $X_{i,j}^{\mu,a}:=\eta_i^{\mu,a} \eta_j^{\mu,a}$, since we consider only different indices for $i$ and $j$, we can say that $\eta_i^{\mu,a}$ and $\eta_j^{\mu,a}$ are i.i.d., therefore, for the CLT, as $M \to + \infty$,
\begin{align}
    \dfrac{1}{M} \SOMMA{a=1}{M} X_{i,j}^{\mu,a} \xrightarrow{\quad d \quad} \bm \bar{X}=\bar{\eta}_i^{\mu} \bar{\eta}_j^\mu
\end{align}
where $\xrightarrow{\quad d \quad}$ expresses the convergence in distribution and $\bar{\eta}$ is the empirical average of all the $\eta_i^{\mu,a}$, which approaches the pattern when $M \to + \infty$. In this way, if we consider $\dfrac{1}{M} \sum_{a=1}^M n_{\mu,a}^2$ and the definition of $n_{\mu,a}$ (see eq.~\eqref{PdO-Uns}), in the $M\to\infty$ limit we can replace $\dfrac{1}{M} \sum_{a=1}^M \eta_i^{\mu,a} \eta_j^{\mu,a}$ with $\R \xi_i^\mu \xi_j^\mu$ so to get $m_\mu^2$ (see eq.~\eqref{eq:mmu}) and, thus, we reach the thesis.
\end{proof}

\begin{corollary}
\label{cor:MEP_UNSUP}
In the asymptotic limit of {\em big data}, namely for $M \to +\infty$,  the Boltzmann-Gibbs  probability distribution $\mathbb{P}_{N,M}(\s \vert \bm \eta)$, related to the unsupervised Hebbian learning Cost function provided in Def. \ref{def:unsup}, approaches the Boltzmann-Gibbs distribution of the Hopfield model, namely
\begin{equation}
    \lim_{M \to \infty} \mathbb{P}_{N,M}(\s \vert \bm \eta) = \mathbb{P}_{N}(\s \vert \bm \xi) =  \frac{1}{Z(\bm \xi)}\exp \left( 
    \dfrac{\beta}{2N} \sum_{\mu=1}^K \sum_{i,j=1}^{N}  \xi_i^{\mu}\xi_j^{\mu}\sigma_i \sigma_j\right).
\end{equation}  
\end{corollary}

\subsubsection{Semi-supervised Hebbian learning}

As we proved convergence of both the supervised and unsupervised protocols to the storage picture painted by Hopfield it is trivial to do the same for the semi-supervised case. We simply report the main definitions and the main corollary, without proving them. 
\begin{definition} (Free energy)
The free energy of the (shallow) semi-supervised Hebbian learning, related to the Cost function provded in Def. \ref{def:semisup}, at finite sizes $N,M$, reads as
\begin{equation}\label{FreeSuper_semisup}
\mathcal{A}_{N,M}(\alpha,\beta) = \dfrac{1}{N} \mathbb{E} \log \sum_{\s} \exp\left( \dfrac{\beta N}{2r^2 M_2} \SOMMA{a=1}{M_2}\SOMMA{\mu=1}{K}\left[\dfrac{s}{\sqrt{(1+\rho_{s})}}\tilde n_\mu +\dfrac{(1-s)}{\sqrt{(1+\rho_{us})}}n_{\mu,a}\right]^2\right).
\end{equation}
where $\mathbb{E}$ is the average w.r.t. $\bm \eta$. 
We remind that the free energy of the Hopfield model at finite size $N$ is 
\begin{align}
    \mathcal{A}_N(\alpha, \beta)= \dfrac{1}{N} \mathbb{E} \log \sums \exp\left( \dfrac{\beta}{2N} \sum_{\mu=1}^K \sum_{i,j=1}^N \xi_i^\mu \xi_j^\mu \si \sigma_j\right),
\end{align}
where $\mathbb{E}$ is now the average w.r.t. $\bm \xi$.
\end{definition}

\begin{definition}{(Interpolating free energy)}\\
Considering a real positive variable $t \in [0,1]$, a dataset $\{ \bm \eta^{\mu,a} \}_{a=1,...,M}^{\mu=1,...,K}$ generated as prescribed in Def. \ref{def:example}, splitting in two parts, made of $M_1$ and $M_2$ items, and the Cost functions of Hebbian storing (see Def. \ref{def:HOP}) and semi-supervised Hebbian learning (see Def. \ref{def:semisup}), the interpolating free energy is introduced as 
\begin{align}\label{interpo_semisup}
    \mathcal{A}_{N,M}(t\vert \alpha,\beta)= \dfrac{1}{N} \mathbb{E} \log \sum_{\s} &\exp \left(t\dfrac{\beta}{2N} \sum_{i,j=1}^{N} \sum_{\mu=1}^K \xi_i^\mu \xi_j^\mu \si \sigma_j\right. \notag \\
    &\left.+ (1-t)\dfrac{\beta N}{2r^2 M_2} \SOMMA{a=1}{M_2}\SOMMA{\mu=1}{K}\left[\dfrac{s}{\sqrt{(1+\rho_{s})}}\tilde n_\mu +\dfrac{(1-s)}{\sqrt{(1+\rho_{us})}}n_{\mu,a}\right]^2\right) \notag \\
    &\hspace{-2.2cm}=\dfrac{1}{N} \mathbb{E}\log Z_{N,M}(t \vert \bm \xi, \bm \eta) = \dfrac{1}{N} \mathbb{E}\log \sum_{\s} B(t \vert \s, \bm \xi, \bm \eta).
\end{align}
where $\R$ as defined in Def. \ref{def:unsup}, $Z_{N,M}(t \vert \bm \xi, \bm \eta)$ is the interpolating partition function, $B(t \vert \s, \bm \xi, \bm \eta)$ the related Boltzmann factor and $\mathbb{E}$ is the expectation w.r.t the variables $\bm \xi$ and $\bm \eta$.
\end{definition}

\begin{theorem}
\label{thm:semisup}
In the asymptotic (\em{big data}) limit, i.e., $M \to +\infty$, the expression of the free energy related to the semi-supervised Hebbian learning approaches that of the Hopfield model.
\end{theorem}

As we have already done for the previous subsections, the proof of Theorem \ref{thm:semisup} is trivial once stated the following 
\begin{lemma}
    In the limit $M \to \infty$, the derivative of the interpolating quenched statistical pressure w.r.t. is null, namely
    \begin{align}
        \dfrac{d\mathcal{A}_{N,M}(t \vert \alpha, \beta)}{dt} = \dfrac{\beta}{2} \sum_{\mu=1}^K \left( \l m_\mu^2 \r - \dfrac{1}{M_2} \sum_{a=1}^{M_2} \left\l \left(\dfrac{s}{\sqrt{(1+\rho_{s})}}\tilde n_\mu +\dfrac{(1-s)}{\sqrt{(1+\rho_{us})}}n_{\mu,a}\right)^2 \right\r\right).
    \end{align}
\end{lemma}

\begin{corollary}
Whatever the value of $s \in [0,1]$, the {\em big data} limit $M \to \infty$ of the above Boltzmann-Gibbs probability distribution, related to (shallow) semi-supervised Hebbian learning, approaches the one related to Hebbian storage (that is the Hopfield model's one):
\begin{equation}
    \lim_{M \to \infty} \mathbb{P}_{N,M}(\s \vert \bm \eta, s) = \mathbb{P}_{N}(\s \vert \bm \xi) = \frac{1}{Z(\bm \xi)}\exp \left( 
    \dfrac{\beta}{2N} \sum_{\mu=1}^K \sum_{i,j=1}^{N} \xi_i^{\mu}\xi_j^{\mu}\sigma_i \sigma_j\right).
\end{equation}  
\end{corollary}

\begin{proof}
    There is no $s$ value that prevents the usage of the CLT on the random variables $\dfrac{1}{M_1} \SOMMA{a=1}{M_1} \eta_i^{\mu,a}$ and $\dfrac{1}{M_2} \SOMMA{a=1}{M_2} \eta_i^{\mu,a} \eta_j^{\mu,a}$, with $i \neq j$, as already done in supervised and unsupervised Sections,  and this suffices to reach the thesis. 
\end{proof}

\subsection{Cost functions vs Loss functions in shallow Hebbian networks}
\label{sec:loss_dense}
As it is pivotal for cross-fertilization to relate observables in Statistical Mechanics with observables in Machine Learning, in this subsection we generalize the correspondence between Cost and Loss functions provided for Hebbian storage (see Remark \ref{rem:loss}) within these learning schemes: we anticipate that the correspondence still holds and, in particular, the Cost function for the supervised protocol will be shown to match the $L^2$ Loss function \cite{Bishop}, while its unsupervised counterpart matches the empirical risk \cite{bojanowski2017unsupervised} .

Let us start with the following definition  
\begin{definition}
    The Loss function of the Hebbian neural network in the supervised setting is 
    \begin{align}
    \mathcal{Ls}^\mu(\bm \sigma \vert \bm \eta^\mu)= Ls_{-}^\mu(\bm \sigma\vert \bm \eta^\mu) Ls_{+}^\mu(\bm \sigma\vert \bm \eta^\mu)
    \end{align}
    where the single pattern $L^2$ Loss function reads as
    \begin{align}
    Ls_{\pm}^\mu(\bm \sigma\vert \bm \eta^\mu) =& \dfrac{1}{2N}  \lVert \dfrac{1}{M}\sum_{a=1}^M \bm{\eta^{\mu,a}} \pm \bm{\sigma} \rVert^2_2 = \notag \\
    =&  \dfrac{1}{2N} \sum_{i=1}^N \left[ \left( \dfrac{1}{M}\sum_{a=1}^M\eta_i^{\mu,a}\right)^2 + \si^2 - 2 \si \dfrac{1}{M} \sum_{a=1}^M \eta_i^{\mu,a} \si \right] = \notag \\
    =& \dfrac{1}{2} \pm \tilde n_\mu + \dfrac{1}{2N}\sum_{i=1}^N \left( \dfrac{1}{M}\sum_{a=1}^M\eta_i^{\mu,a}\right)^2,
\end{align}
so that 
\begin{align}
    \tilde n_\mu^2= \dfrac{1}{4}  + \dfrac{1}{2N} \sum_{i=1}^N \left( \dfrac{1}{M}\sum_{a=1}^M\eta^{\mu,a}\right)^2 + \dfrac{1}{4N^2} \sum_{i,j=1,1}^{N,N} \left( \dfrac{1}{M}\sum_{a=1}^M\eta_i^{\mu,a}\right)^2\left( \dfrac{1}{M}\sum_{a=1}^M\eta_j^{\mu,a}\right)^2-\mathcal{Ls}^\mu(\bm \sigma \vert \bm \eta^\mu).
\end{align}
\end{definition}
We stress that for $M \to \infty$, using CLT, $\dfrac{1}{M} \SOMMA{a=1}{M} \eta_i^{\mu,a}$ can be expressed as a Gaussian variable with null average and $\R/M$ variance. Thus, as $M \to +\infty$, its probability distribution becomes a delta peaked in $\xi_i^\mu$ and $\left(\dfrac{1}{M} \SOMMA{a=1}{M} \eta_i^{\mu,a}\right)^2=(\xi_i^\mu)^2=1$. Therefore, 
\begin{align}
    \tilde n_\mu^2 = 1- \mathcal{Ls}^\mu(\bm \sigma \vert \bm \eta^\mu).
\end{align}

\begin{definition}
The Loss function of the Hebbian neural network in the unsupervised setting is 
    \begin{align}
\mathcal{L}u^{\mu,a}(\bm \sigma \vert \bm \eta^{\mu,a})=Lu_{-}^{\mu,a}(\bm \sigma \vert \bm \eta)Lu_{+}^{\mu,a}(\bm \sigma \vert \bm \eta^{\mu,a})
        \label{eq:lossUNSUP1}
    \end{align}
    where
    \begin{align}
    Lu_{\pm}^{\mu,a}(\bm \sigma \vert \bm \eta^{\mu,a})=\dfrac{1}{2N} \lVert \eta^{\mu,a} \pm \sigma \rVert^2_2 
    = 1\pm  n_{a, \mu}(\bm \sigma \vert \bm \eta^{\mu,a}).
    \label{eq:lossUNSUP2}
    \end{align}
\end{definition}

Using eq.~\eqref{eq:lossUNSUP1} we can rewrite the Cost functions of the Hebbian neural network in supervised and unsupervised settings, respectively eq.~\eqref{eq:H_sup} and eq.~\eqref{eq:H_unsup}, as
\begin{align}
&H_{N,M}^{(Sup)}(\s \vert \bm \eta)= -\dfrac{N }{2\R}\sum_{\mu=1}^K \tilde n_\mu^2 \label{cialtro} \\
&= -\dfrac{N}{2\mathcal{R}} \sum_{\mu=1}^K \left( \dfrac{1}{4} + \dfrac{1}{2N} \sum_{i=1}^N  \left( \dfrac{1}{M}\sum_{a=1}^M\eta_i^{\mu,a}\right)^2 + \dfrac{1}{4N^2} \sum_{i,j=1}^{N}  \left( \dfrac{1}{M}\sum_{a=1}^M\eta_i^{\mu,a}\right)^2 \left( \dfrac{1}{M}\sum_{a=1}^M\eta_j^{\mu,a}\right)^2-\mathcal{Ls}^\mu(\s \vert \bm \eta)\right)  \notag \\
&H_{N,M}^{(Uns)}(\s \vert \bm \eta)= -\dfrac{N }{2 M\R}\sum_{\mu=1}^K \sum_{a=1}^M n_{\mu,a}^2= - \dfrac{N}{2\mathcal{R}} \sum_{\mu=1}^K \dfrac{1}{M} \sum_{a=1}^M \left( 1- \mathcal{L}u^\mu(\s \vert \bm \eta^{\mu,a})\right) \label{eq:HLOSSUNSUP}
\end{align}

\begin{remark}
    Recalling the definition of Empirical Risk $R_\mu^{emp}(\bm \sigma)$ for each pattern $\mu$, which is in this case 
    \begin{align}
        R_\mu^{emp}(\bm \sigma \vert \bm \eta^{\mu,a}) = \dfrac{1}{M} \sum_{a=1}^M \mathcal{L}u^{\mu, a} (\bm \sigma \vert \bm \eta^{\mu,a})
    \end{align}
    we can rewrite it, using eq.~\eqref{eq:lossUNSUP1} and eq.~\eqref{eq:lossUNSUP2}  as
    \begin{align}
        R_\mu^{emp}(\bm \sigma \vert \bm \eta^\mu)= 1- \dfrac{1}{M} \sum_{a=1}^M n_{\mu,a}^2.
    \end{align}
    Therefore, we can write eq.~\eqref{eq:HLOSSUNSUP} as
    \begin{align}
        H_{N,M}^{(Uns)}(\bm \sigma\vert \bm \eta^\mu) = - \dfrac{N}{2\mathcal{R}} \sum_{\mu=1}^K \left( 1- R_\mu^{emp}(\bm \sigma \vert \bm \eta^\mu)\right). 
    \end{align}
and eq.~\eqref{cialtro} as
\begin{align}\label{eccessiva}
        H_{N,M}^{(Sup)}(\bm \sigma\vert \bm \eta^\mu) = -\dfrac{N}{2\mathcal{R}} \sum_{\mu=1}^K \left(  C(\bm \eta^\mu) -\mathcal{Ls}^\mu(\s \vert \bm \eta^\mu)\right),
\end{align}
where $C(\bm \eta )$ is constant in $\bm \sigma$ and thus highlights that all the terms at the r.h.s. of eq.~\eqref{eccessiva} apart $\mathcal{Ls}^\mu(\s \vert \bm \eta^\mu)$ are not functions of neural activities.
\end{remark} 
These calculations show that, in both cases, the strong connection between Cost function in Statistical Mechanics and Loss function in Machine Learning keep holding  and, in particular, all the results derived until now within a statistical mechanical setting by using the Cost function can be obtained also starting from the definition of Loss function within a Machine Learning perspective.

\subsection{Maximum entropy principle for dense neural networks}
So far we inspected only networks whose neuronal interactions happen in couples. While this assumption (the pairwise interactions) seems rather reasonable in modelling biological information processing \cite{Tuckwell}, in Machine Learning's research there is not such a restriction, whose removal can indeed sensibly increase the computational capabilities of the resulting networks \cite{Krotov2018,HopKro1,HopfieldKrotovEffective,unsup,BarraPRLdetective}.
\newline
Roughly speaking, it is reasonable to expect that --as long as the neurons interact in couples (see, e.g. \eqref{eq:H_Hop})-- the Hamiltonians related to the networks are quadratic forms in the order parameters (the $m_{\mu}$ for the Hopfield case, see Def. \ref{def:ordparam}), hence the Boltzmann-Gibbs probabilities that the MEP prescribes to them have Gaussian shape (see e.g. eq.~\eqref{BGmeasure}) and, consistently, solely means and variances in the datasets contribute to their formation (see e.g. eq.~\eqref{constraints}).
\newline
Bypassing the requirement of pairwise interactions, we expect that higher order correlations can be detected by the network\footnote{This is indeed the case and dense networks outperform pairwise neural networks in terms of storage capacity, pattern recognition skills, robustness against adversarial attacks, etc. however we will not deepen network's capabilities in this paper as they can be found elsewhere, see for instance \cite{HopKro1,Krotov2018,HopfieldKrotovEffective,super,unsup,KrotovAlone}.} and this will shine by a glance at the constraints that we impose via the maximum entropy method to achieve their Cost functions, that we provide hereafter. 
\newline
As we did for the pairwise case, we start addressing dense Hebbian storage, then we move toward dense Hebbian learning but, before deepening the derivation of their Cost functions, it is useful to write them explicitly. For dense Hebbian storage \cite{Baldi,Burioni},  dense Hebbian supervised \cite{super} and dense Hebbian unsupervised \cite{unsup} learning these are provided by the next 
\begin{definition}
Let us consider a neural network made of $N$ Ising neurons $\sigma_i \in \{ - 1, +1\}$ for all $i =1,...,N$, the $K$ patterns $\xi^{\mu}$, $\mu =1,...,K$ Rademacher distributed and the standard datasets built of by $M$ blurred examples $\eta^{\mu,a}$, $a =1,...,M$, per pattern. The Hamiltonians of the dense neural networks are the natural many-body extensions of the previously investigated pairwise limits, namely, \fd{for a fixed $P$}
\begin{itemize}
\item    The Cost function of the dense neural network for Hebbian storage is 
    \begin{align}
    \label{eq:H_DHN}
        H_N^{(DHN)}(\bm \sigma \vert \bm \xi) = \dfrac{1}{P! N^{P-1}} \sum_{\mu=1}^K \sum_{i_1, \hdots, i_P=1}^{N} \xi_{i_1}^\mu \cdots \xi_{i_P}^\mu \sigma_{i_1} \cdots \sigma_{i_P}.
    \end{align}
\item   The Cost function of the dense neural network for Hebbian supervised learning is
    \begin{align}
    \label{eq:H_DSN}
    H^{(Sup)}_{N,M}(\boldsymbol{\sigma} \vert {\bm \eta})=& -\dfrac{N}{ P!\mathcal{R}^{P/2} N^P M^P}\SOMMA{\mu=1}{K}\SOMMA{i_1,\hdots,i_P=1}{N}\SOMMA{a_1,\hdots,a_P=1}{M}{\eta_i^{\mu,a_1}}\cdots{\eta_{i_P}^{\mu,a_P}}\sigma_{i_1}\cdots\sigma_{i_P}.
    \end{align}
\item    The Cost function of the dense neural network for Hebbian unsupervised learning is
    \begin{align}
    \label{eq:H_DUN}
    {H}_{N,M}^{(Uns)}(\bm \sigma \vert {\bm \eta}) =& -\dfrac{1}{P! \R^{P/2}\,M\,N^{P-1}}\SOMMA{\mu=1}{K}\SOMMA{a=1}{M}\SOMMA{i_1,\hdots,i_P=1}{N}\eta^{\mu\,,a}_{i_1}\cdots\eta^{\mu\,,a}_{i_P}\sigma_{i_1}\cdots\sigma_{i_P}.
    \end{align}
\end{itemize}
\end{definition}
\red{Note that, while these dense networks can rely upon a massive storage capacity, as $\alpha \propto K^{P-1}/N$, the price to pay is the need for large volumes of examples for training as, e.g. in the worst scenario of unsupervised learning, $M_c \propto 1/r^{2P}$ \cite{unsup}. On the contrary, in the supervised setting, the presence of a teacher allows the standard regime $M_c \propto 1/r^{2}$ \cite{Lucibello}.}

In the next subsections we aim to provide a derivation of all of these Cost functions via the maximum entropy prescription. As in the shallow limit, by inspecting the Lagrange multipliers that proliferate within this approach, this also shed lights on the statistics collected by the network to form its own representation of the patterns\footnote{For the sake of transparency we note that, as long as we work with unstructured datasets (as standard when analytically studying these networks \cite{Coolen}) there is no real need for dense interactions as the Gaussian scenario outlined in the previous Section suffices to collect information to infer the pattern lying behind the provided examples. Yet, at work on structured datasets, the performances of these dense networks are sensibly higher w.r.t. those of the Hopfield model, whatever the protocol (because collecting solely means and variances no longer guarantees a satisfactory probabilistic description of the patterns hidden in the datasets). However we will not deepen their capabilities here, see for instance \cite{super,unsup,Krotov2018,HopKro1,BarraPRLdetective}.}.
 
\subsubsection{Dense Hebbian Storing}
Mirroring the path followed in the previous Section, hereafter we start facing $P-$dense Hebbian storing: this time we have to consider higher order moments. 
\begin{definition}
\label{def:empP}
    In the (dense) Hebbian storing, the empirical averages of the generic $p$-th momentum of the Mattis magnetization $m^p$ with $p=1, \hdots P$ that we need to collect from the dataset read as 
    \begin{align}
        \langle m_\mu^p \rangle_{\textnormal{exp}} = \dfrac{1}{N^p} \sum_{i_1, \hdots, i_p=1}^{N} \xi_{i_1}^\mu\cdots \xi_{i_p}^\mu \langle \sigma_{i_1} \cdots \sigma_{i_p}\rangle_{\textnormal{exp}}
    \end{align}
Note that, in the $p^{th}$ momentum  $m_\mu^p$ related to the generic pattern $\mu$, the $p$-point correlation function among $p$ bits of the pattern is considered.  
\end{definition}
Once defined the empirical quantities we want this dense network to reproduce, we can make the next
\begin{proposition}
\label{prop:MEXDense}
    Given $K$ patterns $\{\boldsymbol{\xi}^{\mu}\}_{\mu =1,...,K}$, generated as prescribed in Def.\ref{DefinitionUno} and $N$ Ising neurons, whose activities read as $\sigma_i \in \{-1, +1 \}$ for all $1,...,N$, the less structured probability distribution $\mathbb{P}_{N}(\bm \sigma \vert \bm \xi)$ that reproduces the empirical average fixed in Def. \ref{def:empP} is 
    \begin{align}
        \mathbb{P}_{N}(\bm \sigma \vert \bm \xi) = \dfrac{1}{Z(\bm \xi)} \exp \left[ \sum_{\mu=1}^K \sum_{p=1}^P \lambda_p^\mu \dfrac{1}{N^p} \left( \sum_{i=1}^N \xi_i^\mu \si \right)^p\right]
    \end{align}
    where $Z(\bm \xi)$ is the partition function and $\lambda_p^\mu, \ \mu=1, \hdots, P$ are the Lagrangian multipliers. 
\end{proposition}

\begin{proof}
    We aim to apply again the maximum entropy procedure. The Lagrange entropic functional suitable to the present case reads as 
    \begin{align}
        S[\mathbb{P}_{N,M}(\bm \sigma \vert \bm \xi)] =& - \sum_{\bm \sigma} \mathbb{P}_{N,M}(\bm \sigma \vert \bm \xi) \log \mathbb{P}_{N,M}(\bm \sigma \vert \bm \xi) + \lambda_0 \left( \sum_{\bm \sigma} \mathbb{P}_{N,M}(\bm \sigma \vert \bm \xi) - 1 \right) \notag \\
        &+ \sum_{p=1}^P \sum_{\mu=1}^K \lambda_p^\mu\left(  \sum_{\bm \sigma} \mathbb{P}_{N,M}(\bm \sigma \vert \bm \xi)  \left( \dfrac{1}{N}\sum_{i=1}^N \xi_i^\mu \si\right)^p - \langle m_\mu^p \rangle_{\textnormal{exp}} \right)
    \end{align}

The stationary conditions of $S[\mathbb{P}_{N,M}(\bm \sigma \vert \bm \xi)]$ w.r.t. $\lambda_0$, $\lambda_p^\mu$ and $\mathbb{P}_{N,M}(\bm \sigma \vert \bm \xi)$ yields
\begin{align}
    \dfrac{\partial S[\mathbb{P}_{N,M}(\bm \sigma \vert \bm \xi)]}{\partial \lambda_0} =& \sum_{\s} \mathbb{P}_{N,M}(\bm \sigma \vert \bm \xi) - 1 =0,\\
     \dfrac{\partial S[\mathbb{P}_{N,M}(\bm \sigma \vert \bm \xi)]}{\partial \lambda_p^\mu} =& \sum_{\bm \sigma} \mathbb{P}_{N,M}(\bm \sigma \vert \bm \xi)  \left( \dfrac{1}{N}\sum_{i=1}^N \xi_i^\mu \si\right)^p - \langle m_\mu^p \rangle_{\textnormal{exp}} =0, \ \ \ p=1, \hdots P,\\
    \dfrac{\delta S[\mathbb{P}_{N,M}(\bm \sigma \vert \bm \xi)]}{\delta \mathbb{P}_{N,M}(\bm \sigma \vert \bm \xi)} =& -\log \mathbb{P}_{N,M}(\bm \sigma \vert \bm \xi) -1 +\lambda_0 + \sum_{\mu=1}^K \sum_{p=1}^P \lambda_p^\mu \dfrac{1}{N^p} \left( \sum_{i=1}^N \xi_i^\mu \si\right)^p =0,      
\end{align}
namely $\mathbb{P}_{N,M}(\bm \sigma \vert \bm \xi)$ can be written as
\begin{align}
\label{eq:PMEP}
    \mathbb{P}_{N,M}(\bm \sigma \vert \bm \xi)= \exp \left( \lambda_0 -1\right) \exp \left[ \sum_{\mu=1}^K \sum_{p=1}^P \lambda_p^\mu \dfrac{1}{N^p} \left( \sum_{i=1}^N \xi_i^\mu \si\right)^p\right], 
\end{align}
so, if we define $\lambda_0$ in such a way that $\exp (1-\lambda_0) = Z( \bm \xi)$, we reach the thesis. 
\end{proof}
 
\begin{corollary}
    If we fix $\lambda_\mu^p$ in Proposition \ref{prop:MEXDense} as 
    \begin{align}
        \lambda^p_\mu = \dfrac{\beta }{P!N^{P-p-1}},
    \end{align}
    we obtain the expression of the probability distribution of the sum of $P$ dense Hebbian neural networks for pattern storage.     
\end{corollary}

\begin{corollary}
    Supposing $P \to +\infty$, if we fix $\lambda_\mu^p$ as 
    \begin{align}
        \lambda_p^\mu= \dfrac{\beta N^{p}}{ p!},
        \label{eq:lambdaexp}
    \end{align}
    we obtain the expression of the probability distribution of the {\em exponential Hopfield model} \cite{demircigil,Lucibello,AllU}. 
\end{corollary} 

\begin{proof}
    If we replace the content of eq.~\eqref{eq:lambdaexp} in eq.~\eqref{eq:PMEP} we find
    \begin{align}
         \mathbb{P}_{N,M}(\bm \sigma \vert \bm \xi) = \exp \left( \lambda_0 -1\right) \exp \left[ \sum_{\mu=1}^K \sum_{p=1}^{+\infty} \dfrac{\beta}{ p!}  \left( \sum_{i=1}^N \xi_i^\mu \si\right)^p\right]
    \end{align}
    but, since we have that 
    \begin{align}
        \sum_{p=0}^{+\infty} \dfrac{1}{p!} \left( \sum_{i=1}^N \xi_i^\mu \si\right)^p = \exp \left( \sum_{i=1}^N \xi_i^\mu \si \right),
    \end{align}
 and   the series is convergent and summable, we reach the thesis. 
\end{proof}

\subsubsection{Dense Hebbian supervised learning}
Along the same lines of the previous subsection, we now generalize the maximum entropy procedure to obtain the Cost function of the supervised Hebbian leaning for $P-$dense neural networks. Let us start by the next
\begin{definition}
\label{def:empirsup}
    In the $P-$dense {supervised Hebbian learning}, the empirical averages of $\tilde n_\mu^p$, $p=1, \hdots, P$ that we need to collect from the dataset read as 
\begin{align}
\label{eq:2.4_2}
    \langle \tilde n_\mu^p \rangle_{\textnormal{exp}} &= \left(\dfrac{r}{\R NM}\right)^p \sum_{i_1,\hdots, i_p=1}^{N}\sum_{a_1, \hdots , a_p=1}^{M} \eta_{i_1}^{\mu,a_1}\eta_{i_p}^{\mu,a_p} \langle \sigma_{i_1} \cdots \sigma_{i_p} \rangle_{\textnormal{exp}},
\end{align}
because the presence of a teacher allows to write the summation over the examples in the above  eq.~\eqref{eq:2.4_2} as performed on  separate indices $a_1, \hdots , a_p$. 
\end{definition}
We can thus give the next
\begin{proposition}
\label{propMEP_supsup}
    Given a dataset $\{ \bm \eta^{\mu,a} \}_{a=1,...,M}^{\mu=1,...,K}$ generated as prescribed in Def. \ref{def:example}, with $r \in [0,1]$ and $\R$ as defined in Def. \ref{def:unsup}, and the $N$ Ising neurons whose neural activities read $\sigma_i = \{- 1, +1\}$ for all $i = 1,...,N$, the less structured probability distribution that reproduces the empirical average fixed in Def. \ref{def:empirsup} is 
    \begin{align}
        \mathbb{P}_{N,M}(\bm \sigma \vert \bm \eta)= \dfrac{1}{Z(\bm \eta)} \exp \left( \sum_{\mu=1}^K \sum_{p=1}^P \lambda_p^\mu \dfrac{1}{N^p}\sum_{i_1, \hdots i_p=1}^{N} \sum_{a_1, \hdots a_p=1}^{M} \eta_i^{\mu,a_1} \cdots \eta_{i_p}^{\mu,a_p} \sigma_{i_1} \cdots \sigma_{i_p} \right),
        \label{eq:propMEP_supsup}
    \end{align}
    where $Z(\bm \eta)$ is the partition function, $\lambda_p^{\mu}$  for all $\mu$ are Lagrangian multipliers, $p=1, \hdots, P$.
\end{proposition}

The proof is similar to the previous one, therefore we omit it. 

\begin{corollary}
    If we fix $\lambda_p^{\mu}$ in Proposition \ref{propMEP_supsup} as 
    \begin{align}
        \lambda_{\mu}^p = \dfrac{\beta \R }{2 r^2 P! N^{P-p-1}}, 
    \end{align}
    we obtain the expression of the probability distribution of the sum of $P$ dense neural networks for supervised Hebbian learning. 
\end{corollary}

\begin{corollary}
\label{cor:S_SUP_dense}
In the asymptotic limit of {\em big data}, namely for $M \to +\infty$,  the Boltzmann-Gibbs  probability distribution $ \mathbb{P}_{N,M}(\s \vert \bm \eta)$,  related to the $P-$dense supervised Hebbian learning (see eq.~\eqref{eq:propMEP_supsup}), approaches the Boltzmann-Gibbs  distribution of the sum of $P$ dense Hebbian neural networks for pattern storage, namely
\begin{equation}
    \lim_{M \to \infty} \mathbb{P}_{N,M}(\s \vert \bm \eta) = \mathbb{P}_{N}(\s \vert \bm \xi) = \frac{1}{Z(\bm \xi)}\exp \left[ \sum_{\mu=1}^K \sum_{p=1}^P \dfrac{\beta}{P!N^{P-1}} \sum_{i_1, \hdots, i_p=1}^{N} \xi_{i_1}^\mu  \cdots \xi_{i_p}^\mu \sigma_{i_1} \hdots \sigma_{i_p}\right].
    \label{eq:Hebb_stordense}
\end{equation} 
\end{corollary}

\begin{proof}
    The proof is rather similar to the one for pairwise neural networks hence we skip it for the sake of {brevity}. 
\end{proof}

\subsubsection{Dense Hebbian unsupervised learning}
Along the same lines of the two previous subsections, we now generalize the maximum entropy procedure to obtain the Cost function of the unsupervised Hebbian leaning for dense neural networks. Let us start by the next
\begin{definition}
\label{def:empirUNSUP}
    In the $P-$dense {unsupervised Hebbian learning}, the empirical averages of $n_{\mu,a}^p$, $p =1, \hdots, P$ that we need to collect from the dataset read as 
\begin{align}
    \langle n_{\mu,a}^p \rangle_{\textnormal{exp}} &= \left(\dfrac{r}{\R N}\right)^p \sum_{i_1, \hdots, i_p=1}^{N} \eta_{i_1}^{\mu,a}\cdots \eta_{i_p}^{\mu,a} \langle \sigma_{i_1} \cdots \sigma_{i_p} \rangle_{\textnormal{exp}}, \ \ \ \mu=1, \hdots, K
\end{align}
where the index $a$ has not been saturated since we lack a teacher that provides this information.
\end{definition}
We can thus give the next
\begin{proposition}
\label{propMEP_UNSUP2}
    Given a dataset $\{ \bm \eta^{\mu,a} \}_{a=1,...,M}^{\mu=1,...,K}$ generated as prescribed in Def. \ref{def:example}, with $r \in [0,1]$ and $\R$ as defined in Def. \ref{def:unsup}, and the $N$ Ising neurons whose neural activities read $\sigma_i \in \{ -1, + 1 \}$ for all $i =1,...,N$, the less structured probability distribution that reproduces the empirical average fixed in Def. \ref{def:empirUNSUP} is 
    \begin{align}\label{etichetta}
       \mathbb{P}_{N,M}(\bm \sigma \vert \bm \eta) = \dfrac{1}{Z(\bm \eta)} \exp \left( \sum_{\mu=1}^K \sum_{a=1}^M \sum_{p=1}^P \lambda_p^{\mu,a} \left( \dfrac{r}{\R N}\right)^p \sum_{i_1, \hdots i_p=1}^{N} \eta_{i_1}^{\mu,a} \cdots \eta_{i_p}^{\mu,a} \sigma_{i_1} \cdots \sigma_{i_p} \right) 
    \end{align}
    where $Z(\bm \eta)$ is the partition function, $\lambda_p^{\mu,a}$  for all $\mu=1, \hdots, K$, {$p=1, \hdots, P$} and $a=1, \hdots, M$ are Lagrangian multipliers.
\end{proposition}
As in the previous subsection, we omit the proof for the sake of {brevity}.

\begin{corollary}
    If we fix $\lambda_p^{\mu,a}$ in Proposition \ref{propMEP_UNSUP2} as 
    \begin{align}
        \lambda_p^{\mu,a} = \dfrac{\beta \R }{2 r^2 M P! N^{P-p-1}},
    \end{align}
    we obtain the expression of the probability distribution of the sum of finite $P$ dense neural networks for unsupervised Hebbian learning. 
\end{corollary}

\begin{corollary}
\label{cor:MEP_UNSUPunsup}
In the asymptotic limit of {\em big data}, namely for $M \to +\infty$,  the Boltzmann-Gibbs  probability distribution $\mathbb{P}_{N,M}(\s \vert \bm \eta)$  related to the $P-$dense unsupervised Hebbian learning (see eq.~\eqref{etichetta}) approaches the Boltzmann-Gibbs distribution of the sum of $P$ dense Hebbian neural networks for pattern storage, namely
\begin{equation}
    \lim_{M \to \infty} \mathbb{P}_{N,M}(\s \vert \bm \eta) = \mathbb{P}_{N}(\s \vert \bm \xi) = \frac{1}{Z(\bm \xi)}\exp \left[ \sum_{\mu=1}^K \sum_{p=1}^P \dfrac{\beta}{P!N^{P-1}} \sum_{i_1, \hdots, i_p=1}^{N} \xi_{i_1}^\mu  \cdots \xi_{i_p}^\mu \sigma_{i_1} \hdots \sigma_{i_p}\right].
    \label{eq:Hebb_stordenseunsup}
\end{equation}   
\end{corollary}

\subsection{Recovering the free energy of dense Hebbian storage in the \em{big data} limit}\label{Guerralike_dense} 

Plan of this Section is to prove  that the whole free energy (and not only the Cost or Loss functions) related to dense Hebbian learning approaches the free energy of the dense Hebbian storage  in the \em{big data} limit: this Section generalizes to many body interactions results collected in Section \ref{Guerralike}. To this task, we rely again upon Guerra's interpolation technique \cite{GuerraNN}.  

\subsubsection{Supervised dense Hebbian learning}

As the ultimate purpose here is to show that the free energy related to dense Hebbian learning collapses on that of dense Hebbian storing, it turns useful to write explicitly the next
\begin{definition}
\label{def:supfree}  (Free energy)
The free energy of the $P-$dense supervised Hebbian learning, whose Cost function is provided in eq.~\eqref{eq:H_DSN},  at finite sizes $N,M$, reads as
\begin{equation}\label{FreeSuper}
\mathcal{A}_{N,M}(\alpha,\beta) = \dfrac{1}{N} \mathbb{E} \log \sum_{\s} \exp\left( \dfrac{\b}{2N^{P-1}\R^{P/2}M^P}\sum_{\mu=1}^K \sum_{i_1,\hdots , i_P=1}^{N} \sum_{a_1, \hdots, a_P=1}^{M} \eta_{i_1}^{\mu,a_1} \cdots \eta_{i_P}^{\mu,a_P} \right),
\end{equation}
where $\b = \dfrac{2\beta}{P!}$ and $\mathbb{E}=\mathbb{E}_{\bm \xi}\mathbb{E}_{(\bm \eta \vert \bm \xi)}$.\\
The free energy of the dense neural network for Hebbian storage, whose Cost function is provided in eq.~\eqref{eq:H_DHN}, at finite size N is 
\begin{align}
    \mathcal{A}_{N}(\alpha,\beta)= \dfrac{1}{N} \mathbb{E} \log \sum_{\s} &\exp \left(\dfrac{\b}{2N^{P-1}} \sum_{\mu=1}^K \sum_{i_1, \hdots, i_P=1}^{N} \xi_{i_1}^\mu \cdots \xi_{i_P}^\mu \sigma_{i_1} \cdots \sigma_{i_P}\right),
\end{align}
where $\mathbb{E}$ is the average w.r.t. $\bm \xi$.
\end{definition}

We can now introduce the following interpolating free energy:
\begin{definition}{(Interpolating free energy)}\\
Considering $t \in [0,1]$, the $K$ patterns as introduced in Def. \ref{DefinitionUno}, their noisy dataset $\{ \bm \eta^{\mu,a} \}_{a=1,...,M}^{\mu=1,...,K}$ generated as prescribed in Def. \ref{def:example}, the Cost functions of dense Hebbian storing (see eq.~\eqref{eq:H_DHN}) and dense supervised Hebbian learning (see eq.~\eqref{eq:H_DSN}), the interpolating free energy is introduced as 
\begin{align}\label{interpo_dense}
    \mathcal{A}_{N,M}(t\vert \alpha,\beta)= \dfrac{1}{N} \mathbb{E} \log \sum_{\s} &\exp \left(t \dfrac{\b}{2N^{P-1}} \sum_{\mu=1}^K \sum_{i_1, \hdots, i_P=1}^{N} \xi_{i_1}^\mu \cdots \xi_{i_P}^\mu \sigma_{i_1} \cdots \sigma_{i_P}\right. \notag \\
    &\left.+ (1-t)\dfrac{\b}{2N^{P-1}\R^{P/2}M^P}\sum_{\mu=1}^K \sum_{i_1,\hdots , i_P=1}^{N} \sum_{a_1, \hdots, a_P=1}^{M} \eta_{i_1}^{\mu,a_1} \cdots \eta_{i_P}^{\mu,a_P} \right) \notag \\
    &\hspace{-2.2cm}=\dfrac{1}{N} \mathbb{E}\log Z_{N,M}(t \vert \bm \xi, \bm \eta) = \dfrac{1}{N} \mathbb{E}\log \sum_{\s} B(t \vert \s, \bm \xi, \bm \eta).
\end{align}
where $\R$ as defined in Def. \ref{def:unsup}, $Z_{N,M}(t \vert \bm \xi, \bm \eta)$ is the interpolating partition function, $B(t \vert \s, \bm \xi, \bm \eta)$ the related Boltzmann factor and $\mathbb{E}$ is the expectation w.r.t the variables $\bm \xi$ and $\bm \eta$.
\end{definition}
We stress that, for $t=0$, the above expression recovers the expression of free energy related to the supervised Hebbian learning, whereas for  $t=1$ it reproduces the expression of the free energy of the dense Hebbian neural network. 

With these premises, we can now state the main theorem of this Section as
\begin{theorem}
\label{prop:limM_supdense}
In the asymptotic (\em{big data}) limit, $M \to +\infty$, the expression of the free energy related to the (dense) supervised Hebbian learning approaches that of the (dense) Hebbian storage.
\end{theorem}

To prove the Theorem, we need to put beforehand the following 
\begin{lemma}
\label{lemma:dert}
In the limit $M \to +\infty$, the derivative of the interpolating free energy w.r.t. $t$ is null, namely 
\begin{align}
   \lim_{M \to \infty} \dfrac{d\mathcal{A}_{N,M}(t\vert \alpha,\beta)}{dt} = \lim_{M \to \infty}  \dfrac{\b}{2} \sum_{\mu=1}^K \left[ \l m_\mu^P \r - \dfrac{1}{\R^{P/2}} \l \tilde n_\mu^{P} \r \right]= 0.
\end{align}
\end{lemma}
\begin{proof}
    Let us start with the computation of the derivative: 
    \begin{align}
        d_t \mathcal{A}_{N,M}(t\vert \alpha,\beta) &= \dfrac{1}{N} \mathbb{E}\Bigg[\dfrac{1}{ Z_{N,M}(t \vert \s, \bm \xi, \bm \eta)}\sum_{\s} B(t \vert \s, \bm \xi, \bm \eta)\left(  \dfrac{\b}{2N^{P-1}} \sum_{\mu=1}^K \sum_{i_1, \hdots, i_P=1}^{N} \xi_{i_1}^\mu \cdots \xi_{i_P}^\mu \sigma_{i_1} \cdots \sigma_{i_P} \right. \notag \\
        &\left.- \dfrac{\b}{2N^{P-1}\R^{P/2}M^P}\sum_{\mu=1}^K \sum_{i_1,\hdots , i_P=1}^{N} \sum_{a_1, \hdots, a_P=1}^{M} \eta_{i_1}^{\mu,a_1} \cdots \eta_{i_P}^{\mu,a_P} \right)\Bigg] \notag \\
        &= \dfrac{\b}{2} \sum_{\mu=1}^K \left(\langle m_\mu^P \rangle - \dfrac{1}{\R^{P/2}} \langle \tilde n_\mu^P \rangle  \right).
        \label{eq:deriv_dense}
    \end{align}
    As already done in the proof of Corollary \ref{cor:S_SUP_dense}, we apply the CLT on the random variable $\dfrac{1}{M} \SOMMA{a=1}{M} \eta_i^{\mu,a}$, therefore we can trace back $\tilde n_\mu$ to $m_\mu$ so to reach the thesis. 
    \end{proof}
    Now we are ready to prove the Theorem.
\begin{proof}(of Theorem \ref{prop:limM_supdense})
We apply the Fundamental Theorem of Calculus over the interpolating variable $t$ while $M \to \infty$: 
\begin{align}
    \mathcal{A}_{N,M}(t=1\vert\alpha,\beta)= \mathcal{A}_{N,M}(t=0\vert \alpha,\beta) + \int_0^1 \dfrac{d \mathcal{A}_{N,M}(t\vert \alpha,\beta)}{dt}\vert_{t=s} ds.
\end{align}
where the commutativity among the \em{big data} limit and the integral over the interpolation parameter is justified by dominated convergence.\\ 
Since for Lemma \ref{lemma:dert} we have that, in the $M \to \infty$, limit, the derivative of the interpolating free energy w.r.t. $t$ becomes zero, this implies 
\begin{align}
    \mathcal{A}_{N}(t=1\vert\alpha,\beta) = \mathcal{A}_{N}(t=0\vert\alpha,\beta),
\end{align}
so the equivalence is proved. 
\end{proof}
Note that this results holds also at finite network sizes $N$, hence the reason for the subscript $N$ in $\mathcal{A}_{N}$.

\subsubsection{Unsupervised Hebbian learning}
Purpose of this Section is to prove that also the free energy related to the (dense) unsupervised Hebbian protocol, in the \em{big data} limit, approaches that of the corresponding (dense) Hebbian storage. Since the strategy we aim to use is quite the same of the previous subsection, we report only the main definitions and results. 
\begin{definition} (Free energy)
The free energy of the $P-$dense unsupervised Hebbian learning, whose Cost function is provided in eq.~\eqref{eq:H_DUN},  at finite sizes $N,M$, reads as
\begin{equation}\label{freeunsup_dense}
\mathcal{A}_{N,M}(\alpha,\beta) =  \dfrac{1}{N} \mathbb{E} \log \sum_{\s} \exp \left(\dfrac{\b}{2 \R^{P/2}\,M\,N^{P-1}}\SOMMA{\mu=1}{K}\SOMMA{a=1}{M}\SOMMA{i_1,\hdots,i_P=1}{N}\eta^{\mu\,,a}_{i_1}\cdots\eta^{\mu\,,a}_{i_P}\sigma_{i_1}\cdots\sigma_{i_P}\right),
\end{equation}
where $\mathbb{E}=\mathbb{E}_{\bm \xi}\mathbb{E}_{(\bm \eta \vert \bm \xi)}$ and $\b$ as in Def. \ref{def:supfree}.\\
The free energy of the dense neural network for Hebbian storage, whose Cost function is provided in eq.~\eqref{eq:H_DHN}, at finite size N is 
\begin{align}\label{zelo1}
    \mathcal{A}_{N}(\alpha,\beta)= \dfrac{1}{N} \mathbb{E} \log \sum_{\s} &\exp \left(\dfrac{\b}{2N^{P-1}} \sum_{\mu=1}^K \sum_{i_1, \hdots, i_P=1}^{N} \xi_{i_1}^\mu \cdots \xi_{i_P}^\mu \sigma_{i_1} \cdots \sigma_{i_P}\right),
\end{align}
where $\mathbb{E}$ is the average w.r.t. $\bm \xi$.
\end{definition}
We can thus introduce the interpolating free energy as stated by the next
\begin{definition}{(Interpolating free energy)}\\
Considering $t \in [0,1]$, the $K$ patterns as introduced in Def. \ref{DefinitionUno}, their noisy dataset $\{ \bm \eta^{\mu,a} \}_{a=1,...,M}^{\mu=1,...,K}$ generated as prescribed in Def. \ref{def:example}, the Cost functions of Hebbian storing (see eq.~\eqref{eq:H_DHN}) and unsupervised Hebbian learning (see eq.~\eqref{eq:H_DUN}), the interpolating free energy is introduced as
\begin{align}
    \mathcal{A}_{N,M}(t\vert \alpha,\beta) = \dfrac{1}{N} \mathbb{E} \log \sum_{\s} &\exp \left(t\dfrac{\b}{2N^{P-1}} \sum_{\mu=1}^K \sum_{i_1, \hdots, i_P=1}^{N} \xi_{i_1}^\mu \cdots \xi_{i_P}^\mu \sigma_{i_1} \cdots \sigma_{i_P}\right. \notag \\
    &\left.+ (1-t)\dfrac{\b}{2 \R^{P/2}\,M\,N^{P-1}}\SOMMA{\mu=1}{K}\SOMMA{a=1}{M}\SOMMA{i_1,\hdots,i_P=1}{N}\eta^{\mu\,,a}_{i_1}\cdots\eta^{\mu\,,a}_{i_P}\sigma_{i_1}\cdots\sigma_{i_P}\right) \notag \\
    &\hspace{-2.2cm}=\dfrac{1}{N} \mathbb{E}\log Z_{N,M}(t \vert \bm \xi, \bm \eta)
\end{align}
where $\R$ as defined in Def. \ref{def:unsup} and $Z_{N,M}(t \vert  \bm \xi, \bm \eta)$ is the interpolating partition function. 
\end{definition}
We stress that for $t=0$ the above interpolating free energy reduces to its expression related to dense unsupervised Hebbian learning  (see  eq.~\eqref{freeunsup_dense}), whereas for  $t=1$ it approaches the expression related to dense Hebbian storing (see eq.~\eqref{zelo1}).

\begin{theorem}
\label{prop:limM_UNSUP_dense}
In the asymptotic (\em{big data}) limit, i.e., $M \to +\infty$, the expression of the free energy related to the dense unsupervised Hebbian learning approaches that of dense Hebbian storing.
\end{theorem}
Mirroring the previous subsection, the proof of Theorem \ref{prop:limM_UNSUP_dense} is trivial provided the next
\begin{lemma}
\label{lemma:dert_UNSUP_dense}
In the limit $M \to +\infty$, the derivative of the interpolating quenched statistical pressure w.r.t. $t$ is null, namely 
\begin{align}
    \dfrac{d\mathcal{A}_{N,M}(t\vert \alpha, \beta)}{dt} = \dfrac{\b}{2} \sum_{\mu=1}^K \left[ 
 \l m_\mu^P \r - \dfrac{1}{M\R^{P/2}} \sum_{a=1}^M \l n_{\mu,a}^P \r \right]\xrightarrow{M\to \infty} 0.
\end{align}
\end{lemma}
\begin{proof}
As already done in supervised setting, we can apply Fundamental Theorem of Calculus and reach the thesis in Theorem \ref{prop:limM_UNSUP_dense}.
\end{proof}

\subsection{Cost functions vs Loss functions in dense Hebbian networks}
\label{sec:loss}
\red{Building on} the analysis conducted for the pairwise scenario, hereafter we deepen the relation between Cost functions preferred in Statistical Mechanics (i.e., Hamiltonians) and Loss functions preferred in Machine Learning (e.g., as the $L_2$ MSE and the empirical risk).
\newline
Let us start with the following definition  
\begin{definition}
    The Loss function of the Dense Hebbian neural network in the supervised setting is 
    \begin{align}
    \mathcal{Ls}^\mu(\bm \sigma \vert \bm \eta^\mu)= Ls_{-}^\mu(\bm \sigma\vert \bm \eta^\mu) Ls_{+}^\mu(\bm \sigma\vert \bm \eta^\mu)
    \end{align}
    where the single pattern average $L^2$ Loss function reads as
    \begin{align}
    Ls_{\pm}^\mu(\bm \sigma\vert \bm \eta^\mu) =& \dfrac{1}{2N^{P/2}} \vert \vert \bm X^{\mu} \pm \bm S \vert \vert_2^2
    \end{align}
    where $\bm X^\mu, \bm S \in \{- 1,+1\}^{N^{P/2}}$ such that 
    \begin{align}
        X_{i_1, \hdots, i_{P/2}}^\mu &\coloneqq \left( \dfrac{1}{M} \sum_{a_1=1}^M \eta_{i_1}^{\mu,a_1} \right) \cdots \left( \dfrac{1}{M} \sum_{a_{P/2}=1}^M \eta_{i_P}^{\mu,a_{P/2}} \right) \\
        S_{i_1, \hdots, i_{P/2}} &\coloneqq \sigma_{i_1} \cdots \sigma_{i_{P/2}}.
    \end{align}
    We can thus write
    \begin{align}
    Ls_{\pm}^\mu(\bm \sigma\vert \bm \eta^\mu)=&  \dfrac{1}{2N^{P/2}} \sum_{i_1, \hdots, i_{P/2}=1,\hdots, 1}^{N, \hdots, N} \left[ \left( X^\mu_{i_1, \hdots, i_{P/2}}\right)^2 + \left( S_{i_1, \hdots, i_{P/2}}\right)^2 - 2 S_{i_1, \hdots, i_{P/2}}  X^\mu_{i_1, \hdots, i_P} \right] = \notag \\
    =& \dfrac{1}{2} \pm \tilde n_\mu^{P/2} + \dfrac{1}{2N^{P/2}}\left( \dfrac{1}{M}\sum_{i=1}^N\sum_{a=1}^M\eta_i^{\mu,a}\right)^{P/2}
\end{align}
such that 
\begin{align}
    \tilde n_\mu^P= \dfrac{1}{4}  + \dfrac{1}{2N^{P/2}}\left( \dfrac{1}{M}\sum_{i=1}^N\sum_{a=1}^M\eta_i^{\mu,a}\right)^{P/2} + \dfrac{1}{4N^P}  \left( \dfrac{1}{M}\sum_{a=1}^M \sum_{i=1}^N \eta_i^{\mu,a}\right)^2\left( \dfrac{1}{M}\sum_{a=1}^M\eta_j^{\mu,a}\right)^{2P}-\mathcal{Ls}^\mu(\bm \sigma \vert \bm \eta^\mu).
\end{align}
\end{definition}
We stress that for $M \to \infty$, using CLT, $\dfrac{1}{M} \SOMMA{a=1}{M} \eta_i^{\mu,a}$ can be expressed as a Gaussian variable equal, on average, to the pattern $\xi_i^\mu$ and with variance $\R/M$, so in this limit $\left(\dfrac{1}{M} \SOMMA{a=1}{M} \eta_i^{\mu,a}\right)^2=(\xi_i^\mu)^2=1$ and
\begin{align}
    \tilde n_\mu^P = 1- \mathcal{Ls}^\mu(\bm \sigma \vert \bm \eta^\mu).
\end{align}

\begin{definition}
The Loss function of the $P-$dense Hebbian neural network in the unsupervised setting 
\cite{bojanowski2017unsupervised} is 
    \begin{align}
\mathcal{L}u^{\mu,a}(\bm \sigma \vert \bm \eta^{\mu,a})=Lu_{-}^{\mu,a}(\bm \sigma \vert \bm \eta)Lu_{+}^{\mu,a}(\bm \sigma \vert \bm \eta^{\mu,a})
        \label{eq:lossUNSUP}
    \end{align}
    where
    \begin{align}
    Lu_{\pm}^{\mu,a}(\bm \sigma \vert \bm \eta^{\mu,a})=\dfrac{1}{2N^{P/2}} \vert \vert \bm X^{\mu,a} \pm \bm S \vert \vert_2^2
    = 1\pm  n^{P/2}_{a, \mu}(\bm \sigma \vert \bm \eta^{\mu,a}).
    \label{eq:lossUNSUP3}
    \end{align}
     where $\bm X^{\mu,a}, \bm S \in \{- 1,+1\}^{N^{P/2}}$ such that 
    \begin{align}
        X_{i_1, \hdots, i_{P/2}}^{\mu,a} &\coloneqq \eta_{i_1}^{\mu,a_1} \cdots \eta_{i_P}^{\mu,a_{P/2}}, \\
        S_{i_1, \hdots, i_{P/2}} &\coloneqq \sigma_{i_1} \cdots \sigma_{i_{P/2}}.
    \end{align}
\end{definition}

Using eq.~\eqref{eq:lossUNSUP3} we can rewrite the Cost functions of the $P-$dense Hebbian neural network in supervised and unsupervised settings, respectively eq.~\eqref{eq:H_DSN} and eq.~\eqref{eq:H_DUN}, as
\begin{align}
&H_{N,M}^{(DSN)}(\s \vert \bm \eta)= -\dfrac{N }{2\R^{P/2}}\sum_{\mu=1}^K \tilde n_\mu^P \notag \\
&= -\dfrac{N}{2\mathcal{R}^{P/2}} \sum_{\mu=1}^K \left(\dfrac{1}{4}  + \dfrac{1}{2N^{P/2}}\left( \dfrac{1}{M}\sum_{i=1}^N\sum_{a=1}^M\eta_i^{\mu,a}\right)^{P/2} + \dfrac{1}{4N^P}  \left( \dfrac{1}{M}\sum_{a=1}^M \sum_{i=1}^N \eta_i^{\mu,a}\right)^2\left( \dfrac{1}{M}\sum_{a=1}^M\eta_j^{\mu,a}\right)^{2P}\right. \notag \\
&\left.-\mathcal{Ls}^\mu(\bm \sigma \vert \bm \eta^\mu)\right) \\
&= -\dfrac{N}{2\mathcal{R}^{P/2}} \sum_{\mu=1}^K \left(  C(\bm \eta^\mu) -\mathcal{Ls}^{\mu}(\s \vert \bm \eta^\mu)\right),\\
&H_{N,M}^{(DUN)}(\s \vert \bm \eta)= -\dfrac{N }{2 M\R^{P/2}}\sum_{\mu=1}^K \sum_{a=1}^M n_{\mu,a}^P= - \dfrac{N}{2\mathcal{R}^{P/2}} \sum_{\mu=1}^K \dfrac{1}{M} \sum_{a=1}^M \left( 1- \mathcal{L}u^{\mu,a}(\s \vert \bm \eta^{\mu,a})\right) \label{eq:HLOSSUNSUP_dense}
\end{align}
where $C(\bm \eta )$ is constant in $\bm \sigma$ and thus highlight that all the terms at the r.h.s. apart $\mathcal{Ls}^{\mu}(\s \vert \bm \eta^\mu)$ are not functions of neural activities.
\newline
We stress that, in plain agreement with the constraints imposed via the maximum entropy extremization, to reach the Cost functions for dense Hebbian learning (namely the matching among the theoretical and empirical $P$-points correlations functions, no longer restricted to the lowest moments $P=1$ and $P=2$),  the present Loss functions force neural outcomes to match empirical ones by considering these higher order correlation functions (rather than just standard quantifiers as means and variances).

\section{Conclusion and outlook} 

In this paper we consistently provided a derivation for all the  Cost functions  (and, in a cascade fashion, for all the Loss functions) that have been recently introduced in the Literature regarding Hebbian learning from examples, with or without the presence of a teacher. To be sharp, we systematically derived, by relying upon the maximum entropy inferential criterion à la Jaynes, the expression for the Cost function of Hebbian storing and Hebbian learning  (supervised, unsupervised and semi-supervised) and then we proved a duality of expression among these Cost functions (preferred by the Statistical Mechanics community) and standard $L^2$ Loss functions (preferred by the Machine Learning community).
\newline
Once provided an exhaustive derivation for the shallow case of neural networks equipped with pairwise interactions among neurons, we enlarged the analysis to the dense counterpart, namely to neural networks with many-body interactions (up to the limit of infinite interactions, that we have shown to collapse on the so-called {\em exponential Hopfield model}).
\newline
Beyond providing an analytical derivation for these expressions (that were solely heuristically assumed in the Literature so far), the usage of the maximum entropy inferential tool allowed also to deepen  how many and what kind of correlations these networks look for when inspecting databases: the higher the order of interactions (i.e. the broader the neural assembly at work), the longer the correlations that the network as a whole can detect.  This possibly contribute to explain why dense neural networks outperform their pairwise shallow limit at work with structured datasets.

\section*{Acknowledgements}
A.B. and F.D. both acknowledge PRIN 2022 grant {\em Statistical Mechanics of Learning Machines: from algorithmic and information theoretical limits to new biologically inspired paradigms} n. 20229T9EAT funded by European Union - Next Generation EU.
\newline
D.P. has been partially supported by the PRIN 2022 grant {\em Elliptic and parabolic problems, heat kernel estimates and spectral theory} n. 20223L2NWK funded by European Union - Next Generation EU. 
\newline
L.A. acknowledges INdAM –GNFM Project (CUP E53C22001930001).
\newline
L.A. and A.B. are members of the group GNFM of INdAM which is acknowledged too.
\newline
F.D. and D.P. are members of the group GNAMPA of INdAM which is acknowledged too.


\begin{thebibliography}{10}

\bibitem{super}
E.~Agliari, L.~Albanese, F.~Alemanno, A.~Alessandrelli, A.~Barra, F.~Giannotti, D.~Lotito, and D.~Pedreschi.
\newblock Dense {H}ebbian neural networks: a replica symmetric picture of supervised learning.
\newblock {\em Physica A: Statistical Mechanics and its Applications}, 626:129076, 2023.

\bibitem{unsup}
E.~Agliari, L.~Albanese, F.~Alemanno, A.~Alessandrelli, A.~Barra, F.~Giannotti, D.~Lotito, and D.~Pedreschi.
\newblock Dense {H}ebbian neural networks: A replica symmetric picture of unsupervised learning.
\newblock {\em Physica A: Statistical Mechanics and its Applications}, 627:129143, 2023.

\bibitem{AABO-JPA2020}
E.~Agliari, L.~Albanese, A.~Barra, and G.~Ottaviani.
\newblock Replica symmetry breaking in neural networks: A few steps toward rigorous results.
\newblock {\em Journal of Physics A: Mathematical and Theoretical}, 53, 2020.

\bibitem{BarraPRLdetective}
E.~Agliari, F.~Alemanno, A.~Barra, M.~Centonze, and A.~Fachechi.
\newblock Neural networks with a redundant representation: Detecting the undetectable.
\newblock {\em Physical Review Letters}, 124:28301, 2020.

\bibitem{EmergencySN}
E.~Agliari, F.~Alemanno, A.~Barra, and G.~De~Marzo.
\newblock The emergence of a concept in shallow neural networks.
\newblock {\em Neural Networks}, 148:232--253, 2022.

\bibitem{Fachechi1}
E.~Agliari, F.~Alemanno, A.~Barra, and A.~Fachechi.
\newblock Generalized {G}uerra's interpolation schemes for dense associative neural networks.
\newblock {\em Neural Networks}, 128:254--267, 2020.

\bibitem{LenkaJPA}
E.~Agliari, A.~Barra, P.~Sollich, and L.~Zdeborova.
\newblock Machine learning and statistical physics: theory, inspiration, application.
\newblock {\em J. Phys. A: Math. and Theor.}, Special, 2020.

\bibitem{Burioni}
E.~Agliari, E.~Barra, R.~Burioni, and A.~Di~Biasio.
\newblock Notes on the p-spin glass studied via {H}amilton-{J}acobi and smooth-cavity techniques.
\newblock {\em Journal of Mathematical Physics}, 53:1--36, 2012.

\bibitem{AgliariDeMarzo}
E.~Agliari and G.~De~Marzo.
\newblock Tolerance versus synaptic noise in dense associative memories.
\newblock {\em European Physical Journal Plus}, 135, 2020.

\bibitem{DenseRSB}
L.~Albanese, F.~Alemanno, A.~Alessandrelli, and A.~Barra.
\newblock Replica symmetry breaking in dense {H}ebbian neural networks.
\newblock {\em J. Stat. Phys.}, 189(2):1--41, 2022.

\bibitem{prlmiriam}
F.~Alemanno, M.~Aquaro, I.~Kanter, A.~Barra, and E.~Agliari.
\newblock Supervised {H}ebbian learning.
\newblock {\em Europhysics Letters}, 141:11001, 2023.

\bibitem{Amit}
D.~J. Amit.
\newblock {\em Modeling brain function: The world of attractor neural networks}.
\newblock Cambridge University Press, 1989.

\bibitem{AGS}
D.~J. Amit, H.~Gutfreund, and H.~Sompolinsky.
\newblock Storing infinite numbers of patterns in a spin-glass model of neural networks.
\newblock {\em Physical Review Letters}, 55:1530--1533, 1985.

\bibitem{Kadmon}
Y.~Bahri, J.~Kadmon, J.~Pennington, S.~S. Schoenholz, J.~Sohl-Dickstein, and S.~Ganguli.
\newblock Statistical mechanics of deep learning.
\newblock {\em Annual Review of Condensed Matter Physics}, 11:1, 2020.

\bibitem{Baldi}
P.~Baldi and S.~S. Venkatesh.
\newblock Number of stable points for spin-glasses and neural networks of higher orders.
\newblock {\em Physical Review Letters}, 58:913, 1987.

\bibitem{Densesterne}
H.~Bao, R.~Zhang, and Y.~Mao.
\newblock The capacity of the dense associative memory networks.
\newblock {\em Neurocomputing}, 469:198--208, 2022.

\bibitem{Jean0}
J.~Barbier and N.~Macris.
\newblock The adaptive interpolation method for proving replica formulas. {A}pplications to the {C}urie--{W}eiss and wigner spike models.
\newblock {\em J. Phys. A: Math. $\&$ Theor.}, 52:294002, 2019.

\bibitem{GuerraNN}
A.~Barra, G.~Genovese, and F.~Guerra.
\newblock The replica symmetric approximation of the analogical neural network.
\newblock {\em Journal of Statistical Physics}, 140:784--796, 2010.

\bibitem{BialekBook}
W.~Bialek.
\newblock {\em Biophysics: searching for principles}.
\newblock Princeton University Press, 2012.

\bibitem{Bishop}
C.~Bishop.
\newblock {\em Pattern recognition and machine learning}.
\newblock Springer Press (New York), 2006.

\bibitem{bojanowski2017unsupervised}
P.~Bojanowski and A.~Joulin.
\newblock Unsupervised learning by predicting noise.
\newblock In {\em International Conference on Machine Learning}, pages 517--526. PMLR, 2017.

\bibitem{Bovier}
A.~Bovier and B.~Niederhauser.
\newblock The spin-glass phase-transition in the {H}opfield model with p-spin interactions.
\newblock {\em Advances in Theoretical and Mathematical Physics}, 5:1001--1046, 8 2001.

\bibitem{Application2}
A.~Cavagna, I.~Giardina, F.~Ginelli, T.~Mora, D.~Piovani, R.~Tavarone, and A.~M. Walczak.
\newblock Dynamical maximum entropy approach to flocking.
\newblock {\em Physical Review E}, 89(4):042707, 2014.

\bibitem{Coolen}
A.~C.~C. Coolen, R.~K\"{u}hn, and P.~Sollich.
\newblock {\em Theory of neural information processing systems}.
\newblock OUP Oxford, 2005.

\bibitem{demircigil}
M.~Demircigil, J.~Heusel, M.~L{\"o}we, S.~Upgang, and F.~Vermet.
\newblock On a model of associative memory with huge storage capacity.
\newblock {\em Journal of Statistical Physics}, 168:288--299, 2017.

\bibitem{Engel}
A.~Engel and C.~Van~den Broeck.
\newblock {\em Statistical mechanics of learning}.
\newblock Cambridge University Press, 2001.

\bibitem{FisherHertz}
K.~Fischer and J.~Hertz.
\newblock {\em Spin Glasses}.
\newblock Cambridge University Press, 1993.

\bibitem{Gardner}
E.~Gardner.
\newblock Multiconnected neural network models.
\newblock {\em Journal of Physics A: General Physics}, 20, 1987.

\bibitem{guerra_broken}
F.~Guerra.
\newblock Broken replica symmetry bounds in the mean field spin glass model.
\newblock {\em Communications in Mathematical Physics}, 233:1--12, 2003.

\bibitem{guerra2002thermodynamic}
F.~Guerra and F.~L. Toninelli.
\newblock The thermodynamic limit in mean field spin glass models.
\newblock {\em Communications in Mathematical Physics}, 230:71--79, 2002.

\bibitem{Hebb}
D.~Hebb.
\newblock New york, ny. the organisation of behaviour (republished 2002), 1949.

\bibitem{Hopfield}
J.~J. Hopfield.
\newblock Neural networks and physical systems with emergent collective computational abilities.
\newblock {\em Proceedings of the National Academy of Sciences of the United States of America}, 79:2554--2558, 1982.

\bibitem{Haiping}
H.~Huang.
\newblock {\em Statistical Mechanics of Neural Networks}.
\newblock Springer press, 2021.

\bibitem{Jaynes}
E.~T. Jaynes.
\newblock Information theory and statistical mechanics.
\newblock {\em Physical Review}, 106:620--630, 1957.

\bibitem{KrotovAlone}
D.~Krotov.
\newblock A new frontier for hopfield networks.
\newblock {\em Nature Review Physics}, (5):366–367, 2023.

\bibitem{Krotov2018}
D.~Krotov and J.~Hopfield.
\newblock Dense associative memory is robust to adversarial inputs.
\newblock {\em Neural Computation}, 30:3151--3167, 2018.

\bibitem{HopfieldKrotovEffective}
D.~Krotov and J.~Hopfield.
\newblock Large associative memory problem in neurobiology and machine learning.
\newblock {\em arXiv}, page 2008.06996, 2020.

\bibitem{HopKro1}
D.~Krotov and J.~J. Hopfield.
\newblock Dense associative memory for pattern recognition.
\newblock {\em Advances in Neural Information Processing Systems}, pages 1180--1188, 2016.

\bibitem{LeCunn}
Y.~LeCun, S.~Chopra, R.~Hadsell, M.~Ranzato, and F.~J. Huang.
\newblock A tutorial on energy-based learning.
\newblock In G.~H. Bakir, T.~Hofmann, B.~Schaelkopf, A.~J. Smola, B.~Taskar, and S.~V.~N. Vishwanathan, editors, {\em Predicting Structured Data}, Neural Information Processing, pages 191--246. The MIT Press, 2007.

\bibitem{Lucibello}
C.~Lucibello and M.~M{\'e}zard.
\newblock Exponential capacity of dense associative memories.
\newblock {\em Physical Review Letters}, 132(7):077301, 2024.

\bibitem{BialekNew}
C.~Lynn, Q.~Yu, R.~Pang, W.~Bialek, and S.~Palmer.
\newblock Exactly solvable statistical physics models for large neuronal populations.
\newblock {\em arXiv}, (2310):2310.10860, 2023.

\bibitem{MPV}
M.~M{\'e}zard, G.~Parisi, and M.~A. Virasoro.
\newblock {\em Spin glass theory and beyond: {a}n {i}ntroduction to the Replica Method and {i}ts applications}, volume~9.
\newblock World Scientific Publishing Company, 1987.

\bibitem{Application3}
T.~Mora, A.~M. Walczak, W.~Bialek, and C.~G. Callan~Jr.
\newblock Maximum entropy models for antibody diversity.
\newblock {\em Proceedings of the National Academy of Sciences}, 107(12):5405--5410, 2010.

\bibitem{Nishimori}
H.~Nishimori.
\newblock {\em Statistical physics of spin glasses and information processing: an introduction}.
\newblock Claredon Press, 2001.

\bibitem{GiorgioOld}
G.~Parisi.
\newblock Infinite number of order parameters for spin-glasses.
\newblock {\em Physical Review Letters}, 43(23):1954, 1979.

\bibitem{GiorgioNobel}
G.~Parisi.
\newblock Nobel lecture: Multiple equilibria.
\newblock {\em Reviews of Modern Physics}, 95.3:030501, 2023.

\bibitem{AllU}
H.~Ramsauer, B.~Sch{\"a}fl, J.~Lehner, P.~Seidl, M.~Widrich, T.~Adler, L.~Gruber, M.~Holzleitner, M.~Pavlovi{\'c}, G.~K. Sandve, et~al.
\newblock Hopfield networks is all you need.
\newblock {\em arXiv preprint arXiv:2008.02217}, 2020.

\bibitem{BialekOld}
E.~Schneidman, M.~Berry, R.~Segev, and W.~Bialek.
\newblock Weak pairwise correlations imply strongly correlated network states in a neural population.
\newblock {\em Nature}, 440(7087):1007, 2006.

\bibitem{storage2}
H.~Seung, H.~Sompolinsky, and N.~Tishby.
\newblock Statistical mechanics of learning from examples.
\newblock {\em Physical review A}, 45(8):1992, 6056.

\bibitem{Application1}
G.~Tka{\v{c}}ik, O.~Marre, T.~Mora, D.~Amodei, M.~J. Berry~II, and W.~Bialek.
\newblock The simplest maximum entropy model for collective behavior in a neural network.
\newblock {\em Journal of Statistical Mechanics: Theory and Experiment}, 2013(03):P03011, 2013.

\bibitem{Tuckwell}
H.~Tuckwell.
\newblock {\em Introduction to theoretical neurobiology}.
\newblock Cambridge University Press, 1988.

\bibitem{SemiSup}
J.~Van~Engelen and H.~Holger.
\newblock Information theory and statistical mechanics.
\newblock {\em Machine Learning}, 109:373--440, 2020.

\end{thebibliography}
\end{document}